\documentclass[lettersize,journal]{IEEEtran}
\usepackage[
n,
operators,
advantage,
sets,
adversary,
landau,
probability,
notions,	
logic,
ff,
mm,
primitives,
events,
complexity,
asymptotics,
keys]{cryptocode}
\usepackage{mdframed, lipsum}
\usepackage{bm}
\usepackage{csquotes}
\usepackage{algorithmic}
\usepackage{algorithm}

\usepackage{amsmath,amsfonts,amssymb,mathtools,amscd}
\usepackage{amsthm}
\usepackage{cleveref}
\usepackage{booktabs}
\usepackage{url}
\usepackage{ifthen}
\usepackage{bm}
\usepackage{multirow}
\usepackage{xcolor,colortbl} 
\usepackage{comment}
\usepackage{paralist,verbatim}
\usepackage{cases}
\usepackage{xargs}
\usepackage{alphalph}
\usepackage{pgfplots}
\usepackage{caption}
\usepackage{subcaption}   
\usepackage{xspace}    
\usepackage{braket} 
\usepackage{cleveref}
\usepackage{color}
\usepackage{numprint}  
\usepackage{lipsum} 
\usepackage[flushleft]{threeparttable}
\usepackage{sidecap} 
\sidecaptionvpos{figure}{c} 

\usepackage{etoolbox}
\patchcmd{\thebibliography}{\chapter*}{\section*}{}{}
	\crefname{theorem}{Theorem}{Theorems}
	\crefname{section}{Section}{Sections}
	\crefname{algorithm}{Algorithm}{Algorithms}	
	\crefname{assumption}{Assumption}{Assumptions}
	\crefname{construction}{Construction}{Constructions}
	\crefname{corollary}{Corollary}{Corollaries}
	\crefname{conjecture}{Conjecture}{Conjectures}
	\crefname{definition}{Definition}{Definitions}
	\crefname{example}{Example}{Examples}
	\crefname{experiment}{Experiment}{Experiments}
	\crefname{counterexample}{Counterexample}{Counterexamples}
	\crefname{lemma}{Lemma}{Lemmata}
	\crefname{observation}{Observation}{Observations}
	\crefname{proposition}{Proposition}{Propositions}
	\crefname{remark}{Remark}{Remarks}
	\crefname{claim}{Claim}{Claims}
	\crefname{fact}{Fact}{Facts}
	\crefname{note}{Note}{Notes}
	\crefname{table}{Table}{Tables}
	\crefname{figure}{Figure}{Figures}
         \crefname{appendix}{Appendix}{Appendices}
	\crefname{equation}{Equation}{Equations}
	
 \crefname{appendix}{App.}{Appendices}
 \crefname{section}{Sec.}{Secs.}
 \crefname{figure}{Fig.}{Figs.}
 \crefname{table}{Tab.}{Tabs.}
 \crefname{remark}{Rem.}{Rems.}
 \crefname{theorem}{Thm.}{Thms.}
 \crefname{definition}{Def.}{Defs.}
 \crefname{equation}{Eq.}{Eqs.}
 \crefname{lemma}{Lem.}{Lems.}

\let\classAND\AND
\let\AND\relax
\usepackage{algorithmic}

\let\AND\classAND
\AtBeginEnvironment{algorithmic}{\let\AND\algoAND}

\let\classOR\OR
\let\OR\relax
\usepackage{algorithmic}

\let\OR\classOR
\AtBeginEnvironment{algorithmic}{\let\OR\algoOR}

\let\classXOR\XOR
\let\XOR\relax
\usepackage{algorithmic}

\let\XOR\classXOR
\AtBeginEnvironment{algorithmic}{\let\XOR\algoXOR}

\let\classNOT\NOT
\let\NOT\relax
\usepackage{algorithmic}

\let\NOT\classNOT
\AtBeginEnvironment{algorithmic}{\let\NOT\algoNOT}

\newtheorem{definition}{Definition}
\newtheorem{theorem}{Theorem}
\newtheorem{proposition}{Proposition}

\renewcommand{\yen}[1]{\textcolor{black}{#1}}

\allowdisplaybreaks 

\DeclareMathAlphabet{\mathpzc}{OT1}{pzc}{m}{it}

\renewcommand{\set}[1]{\{#1\}} 

\renewcommand{\norm}[1]{\lVert #1 \rVert}

\newcommand{\avgcollentropy}{\tilde{\mathbf{H}}_2}
\newcommand{\COL}{\mathtt{COL}} 

\newcommand{\SD}{\mathbf{SD}}  

\newcommand{\Zq}{\mathbb{Z}_q}

\newcommand{\Rn}{\mathbb{R}^n}

\newcommand{\et}{{\it{et al.}}}

\newcommand{\F}{\mathcal{F}}

\renewcommand{\H}{\mathcal{H}}
\renewcommand{\L}{\mathcal{L}}
\newcommand{\M}{\mathcal{M}}

\newcommand{\LB}{\mathcal{L}(\bm{B})}
\newcommand{\LtB}{\mathcal{L}_{\triangle}(\bm{B})}
\newcommand{\ga}{g^\alpha}

    \newcommand{\bB}{\bm{B}}

\newcommand{\Bi}{\bm{B}^{-1}}

\newcommand{\bb}{\bm{b}}
\newcommand{\bc}{\bm{c}}
\newcommand{\be}{\bm{e}}
\newcommand{\bx}{\bm{x}}
\newcommand{\by}{\bm{y}}
\newcommand{\bz}{\bm{z}}

\newcommand{\bu}{\bm{u}}
\newcommand{\bv}{\bm{v}}
\newcommand{\bde}{\bm{\delta}}

\newcommand{\bbx}{\overline{\bm{x}}}
\newcommand{\ox}{\overline{x}}

\def \sample { \overset{_\text{\$}}{\leftarrow} }

\renewcommand{\hash}{\mathsf{H}}
\newcommand{\uhash}{\mathsf{UH}}

\newcommand{\game}{\texttt{Game}}

\renewcommand{\negl}{\mathsf{negl}}

\newcommand{\FNMR}{\ensuremath{\mathsf{FNMR}}\xspace}

\newcommand{\FMR}{\ensuremath{\mathsf{FMR}}\xspace}

\newcommand{\AR}{\mathsf{AR}}
\newcommand{\ARx}{\mathsf{AR}(\bm{x})}
\newcommand{\ARL}{\mathsf{AR}_{\mathcal{L}}}
\newcommand{\ARLx}{\mathsf{AR}_{\mathcal{L}}(\bm{x})}
\newcommand{\Gen}{\mathsf{Gen}}
\newcommand{\Rep}{\mathsf{Rep}}
\newcommand{\RC}{\mathsf{RC}}
\newcommand{\Enc}{\mathsf{Enc}}
\newcommand{\Dec}{\mathsf{Dec}}
\newcommand{\Setup}{\mathsf{Setup}}
\renewcommand{\pp}{\mathsf{pp_{MFFE}}}
\newcommand{\mffe}{\mathsf{MFFE}}
\newcommand{\public}{\mathsf{pp}}
\renewcommand{\vk}{\mathsf{vk}}
\renewcommand{\sk}{\mathsf{sk}}

\newcommand{\comu}{\mathsf{com}_u}
\newcommand{\coms}{\mathsf{com}_s}
\newcommand{\uid}{\mathsf{uid}}
\newcommand{\sid}{\mathsf{sid}}
\newcommand{\UID}{\mathsf{UID}}
\newcommand{\SID}{\mathsf{SID}}
\newcommand{\ru}{\mathsf{ru}}
\newcommand{\rs}{\mathsf{rs}}

\newcommand{\TriLatticeB}{\mathcal{L}_{\triangle}(\bm{B})}

\newcommand{\CV}{\mathsf{CV}}

\newcommand{\CVL}{\mathsf{CV}_{\mathcal{L}}}

\newcommand{\CVLx}{\mathsf{CV}_{\mathcal{L}}(\bm{x})}

\newcommand{\VRL}{\mathsf{VR}_{\mathcal{L}}}

\newcommand{\VRLy}{\mathsf{VR}_{\mathcal{L}}(\bm{y})}

\newcommand{\ConFMR}{\mathsf{ConFMR}}
\newcommand{\ConMFMR}{\mathsf{ConMFMR}}

\newcommand{\DL}{\ensuremath{\mathsf{DL}}\xspace}
\newcommand{\DLsketch}{\ensuremath{\mathsf{DL}^{\sf sketch}}\xspace}

\newcommand{\AdvMfake}{\mathsf{Adv}^{\mathsf{mfake}}_{\mathcal{A}}}

\newcommand{\SuccDLsketch}{\mathsf{Succ}^{\mathsf{DL}^{\mathsf{sketch}}}_{\mathbb{G}, D_B}(T)}
\newcommand{\SuccDL}{\mathsf{Succ}_{\mathbb{G}, D_A}^\DL(T)}
\newcommand{\SuccDLU}{\mathsf{Succ}_{\mathbb{G}, U}^\DL(T)}
\newcommand{\SuccCDH}{\mathsf{Succ}_\mathbb{G}^{\mathsf{cdh}}(T)}
\newcommand{\SuccSig}{\mathsf{Succ}_\mathbb{G}^{\mathsf{sig}}(T)}

\newcommand{\q}{{\texttt{query}}}
\renewcommand{\r}{{\texttt{response}}}

\begin{document}

\title{Biometrics-Based Authenticated Key Exchange with Multi-Factor Fuzzy Extractor}


\author{Hong Yen~Tran,
        Jiankun~Hu,
        ~\IEEEmembership{Senior Member,~IEEE}, 
        and~Wen~Hu,
        ~\IEEEmembership{Senior Member,~IEEE}
}



\maketitle

\begin{abstract}
Existing fuzzy extractors and similar methods provide an effective way for extracting a secret key from a user's biometric data, but are susceptible to impersonation attack: once a valid biometric sample is captured, the scheme is no longer secure. We propose a novel \textit{multi-factor} fuzzy extractor that integrates both a user's secret (e.g., a password) and a user's biometrics in the generation and reconstruction process of a cryptographic key. 
We then employ this multi-factor fuzzy extractor to construct personal identity credentials, which can be used in a new multi-factor authenticated key exchange protocol that possesses multiple important features. First, the protocol provides mutual authentication. Second, the user and service provider can authenticate each other without the involvement of the identity authority. Third, the protocol can prevent user impersonation from a compromised identity authority. Finally, even when both a biometric sample and the secret are captured, the user can re-register to create a new credential using a new secret (\textit{reusable/reissued} identity credentials). Most existing works on multi-factor authenticated key exchange only have a subset of these features. We formally prove that the proposed protocol is semantically secure. 
Our experiments carried out on the finger vein dataset SDUMLA achieved a low equal error rate (EER) of 0.04\%, a reasonable 
computation time of 0.93 seconds for the user and service provider to authenticate and establish a shared session key, and a small communication overhead of 448 bytes. 
\end{abstract}

\begin{IEEEkeywords}
biometrics, finger vein, fuzzy extractor, fuzzy signature, multi-factor, authenticated key exchange.
\end{IEEEkeywords}

\section{Introduction}\label{sec:introduction}

%
%
%
%
\IEEEPARstart{B}{iometrics} authentication \cite{jain2004introduction} refers to the use of a person's unique physiological or behavioural characteristics, such as fingerprints, finger veins, faces, or gaits, to verify his/her identity. While biometric authentication offers many advantages over traditional authentication methods such as secret authentication, there are still some weaknesses of biometric authentication in terms of accuracy and security/privacy \cite{bjorn2005biometrics, georgiev2022fingerprinting, deng2022fencesitter, becker2019poster, geng2019privacy}. First, \textit{false positives} (accepting an unauthorized person as a valid user) and \textit{false negatives} (rejecting a valid user) are major concerns due to the inherent fuzziness of biometric data. Second, \textit{spoofing biometrics} for impersonation is a serious security issue in biometric authentication systems \cite{chingovska2014biometrics}. If an imposter can successfully mimic or replicate biometric traits, they may be able to gain unauthorized access. 
For example, an imposter can recover a fingerprint from a touched object or spoof a 3D face using only one single 2D photo to bypass a fingerprint or 3D face authentication system \cite{wu2022depthfake}. Similar attacks to biometrics-based authentication systems can be found in a number of recent works \cite{li2022video, solano2020scrap, wenger2021hello, wu2022attacks}

Biometric privacy protection and biometric error handling are crucial in creating cryptographic secrets from biometrics for authentication/key exchange protocols. In the literature, some dedicated tools were proposed to deal with these issues, namely fuzzy commitment \cite{juels1999fuzzy}, fuzzy extractor \cite{dodis2008fuzzy}, fuzzy signature \cite{matsuda2016fuzzy, Katsumata2021}, and asymmetric fuzzy encapsulation \cite{wang2021biometrics}. A common drawback when applying these tools to biometrics-based authentication 
\cite{juels1999fuzzy, Katsumata2021, wang2021biometrics} 
is that all of the user's secret random keys tied to the compromised biometrics can be extracted exactly from these tools by the attacker, resulting in impersonation and permanent identity loss. In other words, these schemes are insecure against spoofing attacks and are not reusable when biometrics are compromised (See Table~\ref{tab:mffe-comparison} and B1-Table~\ref{tab:literature-comparison}). 

In biometrics-based authentication, the correct binding of personal identity and biometrics in registration phase is of fundamental importance. To prevent misbinding, individuals should be physically present at a secure place with their identity proof documents (e.g., passport) to provide their biometrics to the identity authority for the registration. If it is not the case (e.g., biometrics are scanned at users' devices without the handling by the identity authority), 
the claimed identity might be misbinded with the biometrics of another individual (e.g., the claimed identity of a victim is paired with the biometrics of an attacker). 
However, implementing this practise might lead to potential user impersonation from privileged insiders as biometrics and other secret data generated in the creation of individual identity credentials are exposed to the identity authority. 
Most existing works do not consider the aforementioned issue and hence often fail to prevent both user identity misbinding and user impersonation attacks from privileged insiders (see~B3~-~Table~\ref{tab:literature-comparison}).

Establishing a secure channel with mutual authenticated key exchange (See B4~-~Table~\ref{tab:literature-comparison}) is more desirable than unilateral authentication in most of client-server applications. In a setting when a user can establish session keys with different servers, many works (e.g., \cite{Yoon2013, He2015, Odelu2015, wazid2017secure, Zhang2019, ma2022outsider, kumar2023robust}) requires a third-party (e.g., a registration center/identity authority, a gateway, a database server) to authenticate users-service providers and establish a session key between them (See B2~-~Table~\ref{tab:literature-comparison}). These schemes generate heavy loads to the identity authority, then raise challenges to build a scalable and reliable system. Therefore, authenticated key exchange schemes without involving the identity authority are more favorable. Besides, nearly all existing works on authenticated key exchange protocols with biometrics were theoretically analysed but lack of experimental evaluation on authentication accuracy with real biometric datasets, making it harder to assess the accuracy of their models (See B5-Table~\ref{tab:literature-comparison}).

 

\begin{table}
\centering
\caption{Comparison of our multi-factor fuzzy extractor
with the related works on the fuzzy extractor and the alike. Key
comparison aspects include A1 - the number of factors in the scheme, A2 - the scheme is secure against spoofing attacks, A3 - the scheme is reusable when all factors are compromised.}
\label{tab:mffe-comparison}
\begin{tabular}{@{}|c|c|c|c|@{}}
\toprule
\textbf{Schemes}        &\textbf{A1}  & \textbf{A2}  & \textbf{A3}  \\ \midrule                          
 Fuzzy commitment \cite{juels1999fuzzy}   &  1                    & $\times$ & $\times$         \\ \midrule
Fuzzy extractor \cite{dodis2008fuzzy}        &  1 & $\times$ & $\times$      \\ \midrule
Fuzzy signature \cite{matsuda2016fuzzy, Katsumata2021}      &  1                    & $\times$ & $\times$         \\ \midrule
Fuzzy encapsulation \cite{wang2021biometrics}        &  1                    & $\times$ & $\times$         \\ \midrule
\textbf{Our multi-factor fuzzy extractor}         & 2 & $\checkmark$ & $\checkmark$                                                             \\ \bottomrule
\end{tabular}
\end{table}

\begin{table}[]
\centering
\caption{Comparison of our multi-factor authenticated key exchange  
with the state-of-the-art protocols. Key comparison aspects include
B1 - user's identity credential can be \textit{reused} for different servers and \textit{reissued} when biometrics are compromised, B2 - user and service provider can authenticate \textit{without} a registration center or a gateway, B3 - can prevent \textit{both} user identity misbinding and user impersonation from a privileged insider (e.g. a compromised registration center), B4 - provide mutual authentication, and B5 - the protocol's authentication accuracy (e.g., EER) was validated on a real biometric dataset.}
\label{tab:literature-comparison}
\begin{tabular}{@{}|c|c|c|c|c|c|@{}}
\toprule
\textbf{Protocols}        & \textbf{B1}  & \textbf{B2}  & \textbf{B3}  & \textbf{B4}  & \textbf{B5} \\ \midrule
\cite{wang2021biometrics}
& $\times$& $\checkmark$ & $\times$& $\checkmark$& $\times$ \\ \midrule
\cite{pointcheval2008multi}   &  $\times$    &  $\checkmark$                     & $\checkmark$                                                & $\times$                              & $\times$         \\ \midrule
\cite{Yoon2013}, \cite{He2015}, \cite{Odelu2015}, \cite{wazid2017secure}, \cite{kumar2023robust}        &   $\checkmark$          & $\times$                      & $\times$                                            & $\checkmark$                               & $\times$         \\ \midrule
\cite{gunasinghe2017privbiomtauth}  & $\checkmark$ &  $\checkmark$   & $\times$                                              & $\times$                                & $\checkmark$        \\ \midrule
\cite{kumari2018provably}, \cite{yang2018cryptanalysis} & $\checkmark$& $\checkmark$ & $\times$ & $\times$& $\times$ \\ \midrule
\cite{Zhang2019}         &     $\checkmark$     &  $\checkmark$                     & $\checkmark$                                        &     $\times$                            & $\times$       \\ \midrule
\cite{zhou2020authentication}, \cite{chuang2021cake}, \cite{barman2018provably} & $\checkmark$& $\checkmark$ & $\times$& $\checkmark$ & $\times$  \\ \midrule
\cite{ma2022outsider} & $\checkmark$ & $\times$ & $\checkmark$ & $\checkmark$ & $\times$ \\ \midrule
\textbf{Our protocol}        &   $\checkmark$           & $\checkmark$                                                                                                             & $\checkmark$                                     & $\checkmark$                                     & $\checkmark$                                                              \\ \bottomrule
\end{tabular}
\end{table}

In this work we first construct a \textit{multi-factor fuzzy extractor} with biometric data, and then integrate it into a multi-factor biometrics-based authenticated key exchange protocol that overcomes all of the aforementioned drawbacks in the existing ones and provides all of the desirable features B1-B5. Note that the state-of-the-art works on authenticated
key exchange with biometrics only have a subset of these features (See Table~\ref{tab:mffe-comparison} and Table~\ref{tab:literature-comparison}). Our proposed protocol is capable of simultaneously dealing with biometrics errors and potential attacks, while maintaining its efficiency. Our main contributions are summarised below.
    \begin{itemize}
        \item We extend the standard fuzzy extractor
        to obtain a novel \textit{multi-factor fuzzy extractor} that employs \textit{both} client's biometrics and her secret to generate a cryptographic key that could be used in 
        an authenticated key exchange. The secret is unknown to anyone except the client, hence providing an extra layer of protection against biometric spoofing attacks. Additionally, the scheme is \textit{reusable} even when all factors are compromised.
        \item Based on the aforementioned multi-factor fuzzy extractor, we develop a new multi-factor biometrics-based authenticated key exchange protocol that provide many desirable security features, addressing the limitations of the state-of-the-art protocols. 
        \item 
        While most existing works did not perform experiments on real biometric datasets, making it harder to assess the accuracy of their models, we implemented and evaluated our protocol on the real finger vein dataset SDUMLA~\cite{yin2011sdumla} with 3,816 finger veins from 106 individuals to validate its accuracy and efficiency. Our protocol achieved a low error equal rate (EER) of 0.04\%, a reasonably small averaged computation time of 0.93 seconds to mutually authenticate a user and a service provider, and a negligible communication overhead of 448 bytes.
    \end{itemize}

The remainder of the paper is organized as follows. In Section~\ref{sec:preliminaries}, we provide necessary concepts and notations used throughout the paper. In Section~\ref{sec:fbpc}, we formally define the multi-factor fuzzy extractor for biometric data and describe a lattice-based instantiation of such an extractor. In Section~\ref{sec:mfake}, we describe and analyse our proposed multi-factor authenticated key exchange protocol. In Section~\ref{sec:experiments}, we discuss the implementation and evaluation of the proposed protocol on the finger vein dataset. The paper is concluded in Section~\ref{sec:conclusion}.

\section{Preliminaries}
\label{sec:preliminaries}

\subsection{Fuzzy Data Setting}
In this section we recall the definition of a fuzzy data setting from Katsumata \et~\cite[Section 2.1]{Katsumata2021}, which describes how fuzzy data can be represented, the metric to measure the closeness of two data points, and how the fuzziness of the data can be quantified. 

A fuzzy data setting $\F=(X,D_X,\AR,E,D_E, \varepsilon)$ consists of the following components.
\begin{itemize}
    \item Fuzzy Data Space $X$ and its distribution $D_X$: $X$ is the space the fuzzy data $\bx$ belongs to, and $D_X$ represents the distribution of the fuzzy data, i.e. $D_X\colon X\to \RR$.
    \item Acceptance Region Function $\AR\colon X \to 2^X$: this function maps $\bx$ to a subset $\ARx\subset X$, so that $\bx'\in X$ is considered to be close to $\bx$ if and only if $\bx'\in \ARx$. The false matching rate (FMR) is defined as $\FMR\triangleq \Pr[\bx,\bx'\gets D_X\colon \bx'\in \ARx]$.
    \item Error Space $E$ and its distribution $D_E$: $E$ is the space of the measurement error and $D_E$ is its distribution, assuming that the measurement error is independent of the individual.
    \item Error Parameter $\varepsilon\in [0,1]$: this parameter defines the false non-matching rate (FNMR), where $\FNMR \triangleq \Pr[\bx \gets D_X, \be \gets D_E\colon \bx + \be \notin \ARx]\leq \varepsilon$. 
\end{itemize}

\subsection{Fuzzy Data Setting with a Lattice} \label{FDS-with-Lattice}

As biometrics can be represented by a vector in $\mathbb{R}^n$, it is convenient to describe a fuzzy data setting based on lattices. 

\begin{itemize}
    \item Let $n \in \NN$ and $\bB \in \RR^{n\times n}$. A lattice $\LB$ spanned by the basis $\bB$ is defined by $\LB \triangleq \{\bB \bz \mid \bz \in \ZZ^n\}$.
    \item For $\bx\in \Rn$ and a lattice $\L=\LB$, the closest vector of $\bx$ in $\L$, denoted $\CVLx$, is a vector $\by \in \L$ satisfying $\norm{\bx-\by}_2\leq \norm{\bx-\bB\bz}_2$ for every $\bz \in \ZZ^n$. Here $\norm{\cdot}_2$ denotes the Euclidean distance in $\Rn$.
    \item The Voronoid region of $\by \in \L=\LB$ is defined as $\VRLy\triangleq \set{\bx\mid \by = \CVLx}$. Note that $\VRLy = \VRL(\bm{0})+\by$.
\end{itemize}

The description of a fuzzy data setting $\F=(X,D_X,\AR,E,D_E,\varepsilon)$ based on a lattice (from Katsumata \et~\cite[Section 5.1]{Katsumata2021}) is given below.
\begin{itemize}
    \item Fuzzy data space is $X=\Rn$ and $D_X$ is its distribution. We associate with $X$ a lattice $\L=\LB$ in which the closest vectors in $\L$ can be efficiently computed.
    \item Acceptance region function $\AR$ is defined as $\ARx=\ARLx\triangleq \set{\bx'\in X\colon \CVL(\bx-\bx')=\bm{0}} = \VRL(\bm{0})+\bx$.
    \item Error space is $E=\Rn$, with $D_E$ is any efficiently samplable distribution over $E$ such that $\FNMR\leq \varepsilon$.
\end{itemize}

Katsumata \et~\cite[Appendix D]{Katsumata2021} also provided a concrete construction of a \textit{triangular} lattice satisfying the above description. 
Apart from the fact that the closest vectors in such a lattice can be efficiently computed, the acceptance regions of triangular lattices are regular \textit{hexagons}, which represent nicely the closeness between vectors in $\Rn$ according to the Euclidean distance.  
The triangular lattice $\L=\TriLatticeB$ with \textit{basis length} $d$ is a lattice spanned by the basis $\bB = \{\bb_1, \bb_2, \cdots, \bb_n\}$ ($\bb_j$'s are columns of $\bB$) satisfying the following two conditions.
\begin{itemize}
    \item $\norm{\bb_i}_2 = d$ for all $i \in [n]$, and
    \item $\langle \bb_i, \bb_j \rangle = d^2 / 2$ for all $i, j \in [n]$ with $i \neq j$.
\end{itemize}
Note that a larger $d$ results in a larger acceptance region.
More specifically\footnote{This construction of $\bB$ was provided to us by the authors of~\cite{Katsumata2021} on email.}, we can choose
$\bb_k = [w_1,  ... , w_{k-1},  r_k, 0,  ... , 0]^{\text{T}}$, where $r_k\hspace{-2pt} = \hspace{-2pt}\sqrt{d^2 \hspace{-2pt}-\hspace{-2pt} \sum_{i=1}^{k-1} w_i^2}$ and $w_k\hspace{-2pt} =\hspace{-2pt} (d^2/2 \hspace{-2pt}-\hspace{-2pt} \sum_{i=1}^{k-1} w_i^2)/r_k$, for $k\in[n]$.
For example, when $n=3$, we have 
\[
\begin{split}
\bb_1 &=(r_1,0,0)=(d,0,0),\\
\bb_2 &= (w_1,r_2,0)=(d/2,\sqrt{3/4}d,0),\\
\bb_3 &= (w_1,w_2,r_3)=(d/2,d/\sqrt{12},\sqrt{2/3}d).
\end{split}
\]

\textbf{Finding the closest vectors.} The closest vector of $\bx\in \Rn$ in a triangular lattice $\LtB$ defined above can be efficiently computed in $O(n^2)$ as illustrated in Algorithm~\ref{alg:CV}. Note that a vector $\bx$ represented with respect to the standard basis in $\Rn$ can be transformed into  $\bbx=(\ox_1,\ldots,\ox_n)=\bB^{-1} \bx$, which is the same vector but represented with respect to the basis $\bB$. Then, $\bx=\bB \bbx=\sum_{i=1}^n\ox_i\bb_i$. 
In our context, it is more convenient and provides better accuracy (avoiding numerical error associated with the multiplication of real numbers) to represent the input $\bx$ and output vectors $\by$ of the closest vector algorithm (Algorithm~\ref{alg:CV}) with respect to the basis $\bB$.
This algorithm is based on the following property of triangular lattice: $x_i\leq x_j$ implies $y_i\leq y_j$ and $x_i \geq x_j$ implies $y_i \geq y_j$, where $\by=(y_1,\ldots,y_n)=\sum_{i=1}^n y_i\bb_i$ is the closest lattice point of $\bx$ in $\LtB$.
Therefore, by first assuming that $x_i\in \mathbb{R}$ and $y_i \in \{0,1\}$ for every $i=1,\ldots,n$, and sorting the coefficient of $\bx$ and $\by$ in a non-decreasing order, we can restrict the list of candidates for the closest lattice point of $\bx$ to just $n+1$ of them, depending on where the zeros and the ones in $\by$ are separated.
Finally, we can translate the lattice point by adding $\lfloor \bx \rfloor$ to it if $x_i\geq 1$ for some $i$.

\begin{algorithm}[htbp]
	\caption{Find the Closest Vector $\CV_{\LtB}(\bx)$ in a Triangular Lattice~\cite[Appendix D]{Katsumata2021}}
	\label{alg:CV}
      \begin{algorithmic}[1]
	\STATE $\bx' = \bm{x} - \lfloor \bm{x} \rfloor$ \quad // so that $\bx'=(x'_1,\ldots,x'_n)$, $0\leq x'_i < 1$
        \STATE Sort the coordinates of $\bx'$ in ascending order, where $\sigma: [n] \rightarrow [n]$ is the permutation representing this sorting
        \STATE $\by^{(1)} = (\underbrace{1,1,\ldots,1}_{n})$
        \STATE $\by^{(n+1)} = (\underbrace{0,0,\ldots,0}_{n})$
		\FOR{$k$ from 2 to $n$}
		\STATE $\by^{(k)} = \by^{(k-1)}$
        \STATE $\by^{(k)}_{\sigma(k-1)} = 0$
		\ENDFOR
		\STATE $\textsf{min} = \textsf{argmin}_{k \in [n+1]}(\norm{\bB(\bx' - \by^{(k)})}_2)$ 
        \RETURN $\by=\by^{(\textsf{min})} + \lfloor \bm{x} \rfloor$ \quad //translate $\by^{(\min)}$, the closest point to $\bx'=\bx-\lfloor \bx\rfloor$, to get the closest point to $\bx$
	\end{algorithmic}
\end{algorithm}

Assuming that the input and output vectors of Algorithm~\ref{alg:CV} are represented with respect to the basis $\bB=\{\bb_1,\ldots,\bb_n\}$ of the 
lattice $\L=\LtB$, then for $\bx'\in \ARL(\bx)$, we have 
\begin{equation} \label{eq:CVL}
\CVL\big(\Bi(\bx-\bx')\big) = \bm{0}.
\end{equation}
Also, for $\bu \in \Rn$, and $\bv \in \ZZ^n$, because $\bB\bv\in \L$, we have
\begin{equation} \label{eq:CVL2}
\CVL(\bu + \bv) = \CVL(\bu) + \bv.
\end{equation}

\subsection{Conditional Collision Entropy, Statistical Distance, Universal Hash Functions and Leftover Hash Lemma}
For a joint distribution $(X, Y)$, \textit{conditional collision entropy} \cite{Katsumata2021} of $X$ given $Y$, denoted as $\avgcollentropy(X|Y) = -\log(\COL(X|Y))$, where 
\[
\COL(X|Y) \triangleq \Pr\left[(x, y), (x', y') \sample (X, Y): x = x'|y = y'\right]
\] 
is called the \textit{conditional collision probability of} $X$ given $Y$.

\textit{Statistical distance} between the two distributions $X$ and $Y$ is $\SD(X, Y) = \frac{1}{2}\sum_v|\Pr[X=v] - \Pr[Y=v]|$.

\textit{Universal hash function and Leftover Hash Lemma} \cite[Lemma A.1]{Katsumata2021}.
$\mathcal{H} = \{\uhash\colon D \rightarrow R \}$ is called a family of universal hash functions
if for all distinct values $x, x' \in D$, $\Pr_{\uhash \leftarrow \mathcal{H}}[\uhash(x) = \uhash(x')] \leq |R|^{-1}$.
Consider the following two distributions.
\[
\begin{split}
D_1 &= \left\{\uhash \sample \mathcal{H}; (x, y) \sample (X, Y): (\uhash, \uhash(x), y)\right\},\\
D_2 &= \left\{\uhash \sample \mathcal{H}; (x, y) \sample (X, Y); r \sample R \colon (\uhash, r, y)\right\}.
\end{split}
\]
Then the Leftover Hash Lemma states that
\[\SD(D_1, D_2) \leq \dfrac{1}{2}\sqrt{|R| \cdot 2^{-\avgcollentropy(X|Y)}} = \dfrac{1}{2}\sqrt{|R| \cdot \COL(X|Y)}.\]


\section{Multi-Factor Fuzzy Extractor} \label{sec:fbpc}

Several dedicated primitives that can protect biometric data and handle biometrics errors have been proposed in the literature, most notably fuzzy commitment~\cite{juels1999fuzzy}, fuzzy extractor~\cite{dodis2008fuzzy}, and fuzzy signature~\cite{Katsumata2021}. 
As pointed out in~\cite{pointcheval2008multi} and other works, a signature or authentication scheme based on these primitives suffer from impersonation attacks as once the biometrics is lost (e.g., the iris captured by a hidden camera, the fingerprint retrieved from a touched object), the attacker can pretend to be the victim without being detected.

In this work, we propose to embed an extra \textit{secret} directly to the fuzzy extractor, which is then used as a primitive to construct a multi-factor AKE (see Section~\ref{sec:mfake}) that is resilient against such an attack.  
Note that a secret can be a \textit{password}, which can be memorized by the user, or a \textit{random string}, which can be stored in a personal device such as the mobile phone.  
As shown in Table~\ref{tab:literature-comparison}, Table~\ref{tab:comp-comparison}, and Table~\ref{tab:comm-comparison}, our new AKE outperforms most previous works in the literature employing the multi-factor approach regarding several security and efficiency aspects.

\subsection{Definition}
\label{subsec:mffe}

We now define a multi-factor fuzzy extractor, which extends the standard concept of fuzzy extractor~\cite{dodis2008fuzzy}.

\begin{definition}[Multi-factor fuzzy extractor]\label{def:mffe}
An $(\F,A,D_A,B,\epsilon, \Delta)$ multi-factor fuzzy extractor (MFFE), where 
\begin{itemize}
    \item $\F\hspace{-3pt}=\hspace{-3pt}(\hspace{-2pt}X,\hspace{-2pt}D_X,\hspace{-2pt}\AR,\hspace{-2pt}E,\hspace{-2pt}D_E,\hspace{-2pt}\varepsilon\hspace{-2pt})$ represents the fuzzy data setting, 
    \item $A$ and $D_A$ 
    are the space of the secret and its distribution, 
    \item $B$ is space of the (secret) cryptographic key extracted from the biometric data and the secret, 
    \item $\epsilon$ represents the upper bound on the statistical distance between the key and the uniform bit string of 
    equal length,
    \item $\Delta$ is the sketch space,
\end{itemize}
consists of a deterministic procedure $\Setup()$ and two randomized procedures $\Gen()$ and $\Rep()$ with the following properties.
\begin{itemize}
    \item $\Setup(\F,A,B)\to \pp$: the procedure $\Setup$ takes as input the fuzzy data setting $\F$, the space of the secret $A$, and the space of the cryptographic key $B$, and generates the public parameters $\pp$ for the extractor. 
    \item $\Gen(\pp,\bx,\ga) \to (\beta,\bde, w)$: the generation procedure $\Gen$ takes as input the public parameters, a biometric data $\bx\in X$, the binding $g^\alpha$ of a secret $\alpha\in A$, and generates a cryptographic key $\beta\hspace{-2pt}\in\hspace{-2pt} B$, a sketch $\bde\hspace{-2pt} \in\hspace{-2pt} \Delta$, and helper data $w$.
    \item $\Rep(\pp, \bx', \alpha, \bde, w) \to \beta$ if $\bx'\in \AR(\bx)$: the reproduction procedure $\Rep$ takes as input the public parameters, a biometric data $\bx'\in X$, a secret $\alpha\in A$, a sketch $\bde\in \Delta$, helper data $w$, and outputs $\beta$ if $\bx'$ lies in the acceptance region of $\bx$.
\end{itemize}
\end{definition}

Let $\kappa\in \mathbb{N}$ denote a security parameter, and $\negl(\kappa)$ the set of \textit{negligible functions} in $\kappa$.
A positive-valued function $\varepsilon(\kappa)$ belongs to $\negl(\kappa)$ if for every $c>0$, there exists a $\kappa_0\in \mathbb{N}$ such that $\varepsilon(\kappa)<1/\kappa^c$ for all $\kappa > \kappa_0$.

\begin{definition}[\DLsketch assumption]
    Given a fuzzy data setting $\F=(X,D_X,\AR,E,D_E, \varepsilon)$, a multi-factor fuzzy extractor $(\F,A,D_A,B,\epsilon, \Delta)$ and a group $\GG$ with a generator $g$, we say that the discrete logarithm with sketch ($\DLsketch$) assumption holds if for all PPT adversaries $\mathcal{A}$, the following probability, denoted by $\SuccDLsketch$, belongs to $\negl(\kappa)$, where $\kappa$ is the security parameter.
    \begin{equation*}
        \Pr\hspace{-2pt} \left[
        \begin{array}{cc} \hspace{-5pt}
            \hspace{-3pt} \pp \leftarrow \Setup(\F,A,B) \\
            \hspace{-3pt}\bx \sample D_X, \alpha \sample D_A \\
            \hspace{-5pt}(\beta, \bde, w)\hspace{-2pt} \leftarrow\hspace{-2pt} \Gen(\pp, \bx, \ga)
        \end{array}\hspace{-10pt}
        : \mathcal{A}(\pp, g^\beta, \bde, w)\hspace{-2pt} =\hspace{-2pt} \beta \hspace{-1pt}
        \right]\hspace{-2pt}.
    \end{equation*}
\end{definition} \label{def:dlsketch}

The $\DLsketch$ assumption means that the DL problem, which reconstructs $\beta$ from $g^\beta$, is still a hard problem given the sketch $\bde$ and the helper data $w$. 
This assumption would imply the security property of MFFE (See Proposition \ref{pro:mffe_security}, Section \ref{subsec:mffe-lattice}), that is, it is difficult for the adversary who has obtained the sketch and the helper data $(\bde, w)$ to distinguish the extracted key $\beta$ from a uniform random value in the key space $B$ as their statistical distance is negligible. 


\subsection{A Multi-Factor Fuzzy Extractor with a Lattice} \label{subsec:mffe-lattice}

We describe below a lattice-based $(\F,A,D_A,B,\epsilon, \Delta)$-MFFE.
Note that in our case, the deep learning model feature extractor in~\cite{kuzu2021loss} produces feature vectors $\bx \in \mathbb{R}^n$ from raw biometric data.
Hence, we can assume $X=\mathbb{R}^n$. The correctness of the lattice-based MFFE is proved in Proposition~\ref{pro:mffe_correctness}.
\begin{itemize}
    \item The fuzzy data setting $\F$ has $(X=\mathbb{R}^n$, $\AR=\ARL$, where $\L=\LtB$ is a triangular lattice, $E=\Delta=\Rn$.
    \item The secret and key spaces $A$ and $B$ are both $\Zq$ for a large prime $q$.
    \item $\Setup(\F,A,B)\to \pp = \set{\bB, \GG, g, \H}$, where $\GG$ is a group of prime order $q$ with a generator $g$, 
    and $\H=\set{\uhash\colon \Zq^n\to \Zq}$ is a family of universal hash functions.
    \item $\Gen(\pp, \bx, \ga) \rightarrow (\beta, \bde, w)$ works as follows. First, compute $\bc = \lfloor \Bi \bx\rfloor || \hash(\ga)$ where $\hash$ is a collision-resistant hash function modulo $q$. Select a random hash function $\uhash \sample \H$ and calculate 
    $\beta = \uhash(\bc)$. Here, $\uhash(\bc)$ computes the inner-product of $\bc$ and a random $(n\hspace{-2pt}+\hspace{-2pt}1)$-dimensional vector in $\Zq^{n+1}$ modulo $q$.
    Sample $r \sample \Zq$ and calculate $w = g^r, \bde = \Bi\bx - \Enc(k, \lfloor \Bi\bx \rfloor)$ where $\Enc(k, \lfloor \Bi\bx \rfloor)$ is a symmetric encryption, e.g. AES-256 of $\lfloor \Bi\bx \rfloor$ 
    using the key $k = \hash(w, (\ga)^r)$. 
    \item $\Rep(\pp, \bx', \alpha, \bde, w) \rightarrow \beta$. First, recover the key $k = \hash(w, w^\alpha)$ using the secret $\alpha$ and the helper data $w$. Compute $\CVL(\Bi\bx' - \bde)$, then recover $\bc = \Dec(k, \CVL(\Bi\bx' - \bde))||\hash(g^\alpha)$. 
    Return $\beta=\uhash(\bc)$.
\end{itemize}

Compared to fuzzy signature and linear sketch in Katsumata \et~\cite{Katsumata2021}, the introduction of $\Enc(k, \lfloor \Bi\bx \rfloor)$ in computing sketch $\bde$ is to include the second factor $\alpha$ in the form of the encryption/decryption key $k$. This modification makes the reconstruction of $\beta$ from $\bde$ harder for the adversary, requiring her to obtain both a close biometrics and $\alpha$.

\begin{proposition}[MFFE Correctness]
\label{pro:mffe_correctness}
In the lattice-based $(\F,A,D_A,B,\epsilon, \Delta)$-MFFE described above, let $(\beta, \bde, w) \gets \Gen(\pp, \bx, \ga)$ and assume that $\bx'\in \ARL(\bx)$. Then $\Rep(\pp, \bx', \alpha, \bde, w)=\beta$.
\end{proposition}
\begin{proof}
As $\alpha$ is the secret used in $\Gen(\cdot)$, $k = \hash(w, w^\alpha) = \hash(w, (g^r)^\alpha) = \hash(w, (\ga)^r)$.
If $\bx'\in \ARL(\bx)$, 
then using \eqref{eq:CVL} and \eqref{eq:CVL2}, we have:
\begin{align*}    
    &\CVL(\Bi \bx'-\bde) \\ &= \CVL \left(\Bi\bx' - \left(\Bi\bx - \Enc(k, \lfloor \Bi\bx \rfloor) \right)\right) \\
    &= \CVL\left(\Bi(\bx'-\bx)+ \Enc(k, \lfloor \Bi\bx \rfloor) \right) \\
    &=  \Enc(k, \lfloor \Bi\bx \rfloor)  .
\end{align*}
Therefore,
\begin{align*}
    &\uhash(\Dec(k, \CVL(\Bi\bx' - \bde))||\hash(g^\alpha))
    \\
    = &\uhash(\Dec(k, \Enc(k, \lfloor \Bi\bx \rfloor)||\hash(g^\alpha)) 
    \\
    =& \uhash(\lfloor \Bi\bx \rfloor||\hash(g^\alpha)) = \beta, 
\end{align*}
as claimed.
\end{proof}

The security property of MFFE is captured in Proposition~\ref{pro:mffe_security}, which states that the extracted cryptography key $\beta$ is nearly uniform even for the adversaries who observe the sketch $\bde$.

\begin{proposition}[MFFE Security]
\label{pro:mffe_security}
Given the lattice-based $(\F,A,D_A,B,\epsilon, \Delta)$-MFFE described above and $(\beta, \bde, w) \gets \Gen(\pp, \bx, \ga)$, if the conditional multi-factor false matching rate ($\ConMFMR$) 
belongs to $\negl(\kappa)$, 
then\footnote{Note that since $w$ and $\beta$ are independent, knowing $w$ does not give the adversary any additional information about $\beta$.}
\[\SD((\beta, \bde),(b \sample B, \bde)) \in \negl(\kappa),\] 
where 
$\ConMFMR$ is defined as
\begin{equation}
    \ConMFMR \triangleq \Pr \left [
        \begin{array}{cc}
            (\bx, \alpha), (\bx', \alpha') \leftarrow (D_X, D_A), \\
            \bde \leftarrow \Gen(\pp, \bx, \ga), \\
            \bde' \leftarrow \Gen(\pp, \bx', g^{\alpha'}) :\\
            \bx' \in \AR(\bx), \alpha = \alpha' | \bde = \bde'
        \end{array}
        \right ].
\end{equation}
\end{proposition}

\begin{proof}
According to the Leftover Hash Lemma \cite[Lemma A.1]{Katsumata2021}, we have
\[
\begin{split}
    & \SD\left((\beta, \bde),(b \sample B, \bde)\right) \\ 
    = & \SD\left(\left(\uhash, \uhash(\lfloor \Bi\bx\rfloor||\hash(\ga)), \bde\right),(\uhash, b \sample B, \bde)\right) \\
    \leq & \dfrac{1}{2}\sqrt{|B| \cdot \COL\big(\lfloor \Bi\bx\rfloor || \hash(\ga) | \bde\big)}. 
\end{split}
\]
Moreover, 
\begin{align*}
    &\COL\left(\lfloor \Bi\bx\rfloor || \hash(\ga) | \bde\right) \\
    =& \Pr \left[
        \begin{array}{cc}
            (\bx, \alpha, \bde), 
            (\bx', \alpha', \bde') \leftarrow (X, A, \Delta): \\
            \lfloor \Bi\bx\rfloor || \hash(\ga) = \lfloor \Bi\bx'\rfloor || \hash(g^{\alpha'}) | \bde = \bde'
        \end{array}
        \right] \\
    =& \Pr \left[
        \begin{array}{cc}
            (\bx, \alpha, \bde), (\bx', \alpha', \bde') \leftarrow (X, A, \Delta) :\\
            \lfloor \Bi\bx\rfloor = \lfloor \Bi\bx'\rfloor, \alpha = \alpha'| \bde = \bde'
        \end{array}
        \right] \\
    \leq & \Pr \left[
        \begin{array}{cc}
            (\bx, \alpha, \bde), (\bx', \alpha', \bde') \leftarrow (X, A, \Delta) :\\
            \Enc(k, \lfloor \Bi\bx\rfloor)  = \Enc(k, \lfloor \Bi\bx'\rfloor), \alpha = \alpha' | \bde = \bde'
        \end{array}
        \right] \\
    = & \Pr \left[
        \begin{array}{cc}
            (\bx, \alpha, \bde), (\bx', \alpha', \bde') \leftarrow (X, A, \Delta) :\\
            \Bi\bx  = \Bi\bx', \alpha = \alpha'|
            \bde = \bde'
        \end{array}
        \right] \\
    = & \COL\left(\left(\Bi\bx, \alpha\right) | \bde\right) \\
    < & \Pr \left [
        \begin{array}{cc}
            (\bx, \alpha), (\bx', \alpha') \leftarrow (X, A) \\
            \bde \leftarrow \Gen(\pp, \bx, \ga) \\
            \bde' \leftarrow \Gen(\pp, \bx', g^{\alpha'}) :\\
            \bx' \in \AR(\bx), \alpha = \alpha' | \bde = \bde'
        \end{array}
        \right ]
    = \ConMFMR.
\end{align*}
From here, using the same argument as in~\cite[Sec. 5.4]{Katsumata2021} but replacing their \textit{conditional false matching rate} $\ConFMR$ (i.e., the probability of conditional matching of \textit{biometrics factor only}) by our \textit{conditional multi-factor false matching rate} $\ConMFMR$ (i.e., the probability of conditional matching of \textit{both biometrics and secret factors}), we have 
\begin{itemize}
    \item[(1)] if 
    $\ConMFMR \lessapprox 2^{-(2\kappa+\omega(\log\kappa))}$ then the standard \DL assumption implies \DLsketch assumption;
    \item[(2)] if 
    $\ConMFMR \approx 2^{-\kappa}$ (or $2^{-\omega(\log \kappa)}$) then the \DLsketch assumption holds in the generic group model.
\end{itemize}
In the other words, the statistical distance of the extracted key $\beta$ and a uniformly random variable is negligible. 
\end{proof}


\subsection{Comparisons to Related Works on Fuzzy Extractor and Its Similar Schemes} \label{subsec:comparision-mffe}

\textit{Fuzzy commitment}~\cite{juels1999fuzzy} with the code-offset construction was the first approach dealing with the fuzziness of biometrics utilising error-correcting codes and cryptography. The idea is to use a biometric template $\bx$ to commit to a secret uniform random codeword $\bc$ without revealing the codeword. This codeword can only be restored by using a 'close-enough' biometric sample $\bx'$ and an error correction code using the distance between that codeword and the original biometric template ($\bde = \bx - \bc$). However, once the adversary obtains a valid biometric sample, she can always recover the secret codeword $\bc$ even when the user selects a new one, rendering the scheme unusable. 

\textit{Fuzzy extractor} was proposed in \cite{dodis2008fuzzy} to extract a nearly uniform cryptographic key $\beta$ from a biometrics $\bx$, which can then be recovered by using another biometrics $\bx'$ close to $\bx$ (with respect to metrics such as the Hamming distance and the edit distance) and a publicly available sketch $\bde$. 
As $\beta$ remains nearly uniform even given $\bde$, it can be used as a secret key in a cryptographic application without being stored as it can always be recovered from a valid biometric sample $\bx'$ close enough to $\bx$. 
Hence, if the adversary captures a valid biometric sample, then $\beta$ can always be recovered. Thus, once the biometrics is lost, the extractor can no longer generate a safe key.

\textit{Fuzzy signature} was recently introduced in \cite{Katsumata2021} to sign a message using a user's biometrics $\bx$. A linear sketch $\bde$ over a lattice is used to extract a cryptographic key $\beta$ from the user's biometrics, which is then fed as a \textit{signing} key to the standard Schnorr signature scheme to produce a fuzzy signature. The signature can be verified at a remote server using a \textit{verification} key generated earlier from the user's enrolled biometric template. The subtle difference with a fuzzy extractor~\cite{dodis2008fuzzy} is that thanks to the linearity of the sketch, given a valid biometric sample, the verification works \textit{without} requiring the recovery of the exact key produced earlier in the keys generation process. 
However, similar to fuzzy extractors, once a valid biometric sample is captured by an adversary, the fuzzy signature scheme is no longer secure.

With the same purpose of protecting the privacy of a user's biometrics while tolerating noises for a bidirectional biometrics authenticated key exchange, an asymmetric \textit{fuzzy encapsulation} mechanism is proposed in \cite{wang2021biometrics}. In this scheme, the user derives a biometric secret key directly from her scanned biometrics and generates a public key to encapsulate messages which can be decapsulated only by a similar biometric secret key. Similar to the above schemes, once a biometric sample is lost, the scheme is no longer secure.

All of the aforementioned schemes suffer from the \textit{spoofing} attack: if the user's biometric data is compromised, all the keys generated from the biometrics can be reproduced by the attacker. 
For example, an adversary can successfully impersonate a user by using their stolen biometrics to generate a valid fuzzy signature~\cite{Katsumata2021} or to decapsulate a secret message~\cite{wang2021biometrics}.
This is due to the fact the keys in such schemes are generated from the biometric data only. 

\textit{Our proposed multi-factor fuzzy extractor} (MFFE) addresses this gap by embedding an additional factor (a secret/password) directly into a fuzzy extractor~\cite{dodis2008fuzzy}. 
This MFFE prevents the adversary from reproducing the cryptographic keys even when the biometric data is leaked, as long as the secret (password) remains secure. Similarly, the adversary cannot reconstruct the keys when the secret is leaked but the biometric data remains unknown. 
In particular, in the event that \textit{both biometric data and secret are compromised}, the user can simply use a new secret to prevent the adversary from obtaining the keys in the future executions of MFFE. By contrast, biometrics-only schemes such as~\cite{juels1999fuzzy, dodis2008fuzzy, Katsumata2021, wang2021biometrics} can no longer be used in the future once the biometric data is lost.  
This feature is particularly useful when MFFE is employed in a multi-factor authenticated key exchange, which \textit{enables a recovery option} for the user when all the factors are compromised.  

\section{An Authenticated Key Exchange Based on Multi-Factor Fuzzy Extractor}
\label{sec:mfake}

We developed in this section a new authenticated key exchange (AKE) protocol based on the multi-factor fuzzy extractor introduced in Section~\ref{sec:fbpc}. 
Unlike most existing password-authenticated key exchange (PAKE) protocols, we provide a complete protocol, which consists of three components: biometrics extraction and registration (using our multi-factor fuzzy extractor), password creations (between the user and the service provider, using our multi-factor fuzzy extractor), and a standard PAKE. Our protocol is ready to be deployed on any biometric dataset. We first give an overview of the protocol in Section~\ref{sec:system-overview}, then provide a full description in Section~\ref{sec:mfake-description}. The correctness and security proofs can be found in Section~\ref{sec:mfake-analysis}. Comparisons with related works are provided in Section~\ref{sec:mfake-relatedworks}.
\subsection{System overview} \label{sec:system-overview}

\begin{figure}[htb!]
\centering
\includegraphics[scale=0.87]{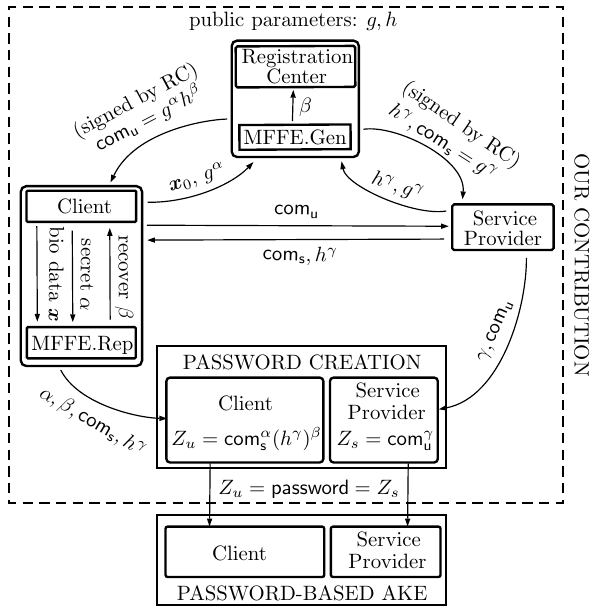}
\caption{An overview of our multi-factor authenticated key exchange protocol, which is based on 
a novel multi-factor fuzzy extractor (MFFE) that allows a client to use a secret $(\alpha)$ in combination with her bio-data $(\bx\text{ and } \bx_0)$ to generate a cryptographic key $(\beta)$. The secret is also involved in the creation of a common password between the client and a service provider, which can then be employed in a standard password-based authenticated key exchange~\cite{abdalla2005simple}. The presence of the secret prevents an impersonation attack from a semi-honest registration centre, who knows the biometric data $\bx_0$, and also makes the protocol reusable from biometric and secret leakage. }
\label{fig:overview}
\end{figure}

Our scheme involves three parties: User - the client party that tries to be authenticated with multi-authentication factors, Service Provider (SP) - the server party that authenticates the user before allowing the user to perform any transactions, and Registration Center (RC) - the identity authority party that approves service providers and individuals' identities. 
The following assumptions have been made:1- The registration and update/revocation phases happen via a secure channel, for example, users present in person at RC's secure location with their physical identity documents (e.g., passport) and communicate with RC via a \textit{secure} channel established at that time. 2- In the authentication phase, users and remote servers communicate via \textit{insecure} public channels;
We consider a malicious outsider adversary who has a complete control over the communication channel in the authentication phase; thus, she is capable of intercepting, modifying, and replaying messages transmitted over this channel. This adversary is able to steal a user's biometrics or access a user's device to obtain the stored secret but cannot corrupt all authentication factors. We also consider a semi-honest insider adversary who can corrupt RC to record data generated during the registration process and eavesdrop on the communication channel in the authentication phase. 

We aim to achieve a mutually authenticated key exchange which satisfies semantically security 
and is robust to impersonation from privileged insider attacks. Unless all authentication factors have been corrupted, it is difficult for an adversary to impersonate a client to deceive the server or vice versa. All the keys agreed on between a client and a server are semantically secure.

\begin{figure*}
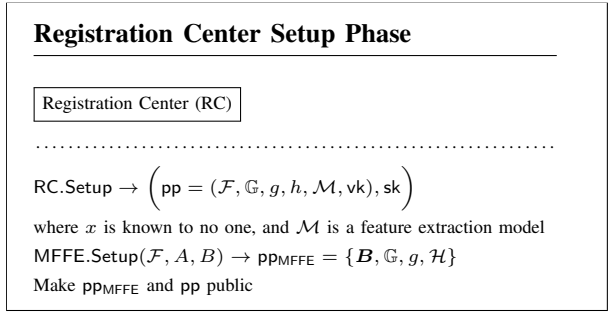

    \centering
    \begin{mdframed}[userdefinedwidth=8cm,align=center]
        \procedure[codesize=\scriptsize]{\textbf{Registration Center Setup Phase}}{
            \\
	\framebox{Registration Center (RC)} \< \< \< \<    
        \pclb
	\pcintertext [dotted] {}
        \RC.\Setup \to \bigg(\public=(\F, \GG, g, h, \M, \vk), \sk\bigg) \\
        \text{where $x$ is known to no one, and } \M \text{ is a feature extraction model } \< \< \< \< \\
        \mffe.\Setup(\F,A,B)\to \pp = \set{\bB, \GG, g, \H} \< \< \< \< \\
        \text{Make } \pp \text{ and } \public \text{ public} \< \< \< \< 
	}
	\end{mdframed}
    \caption{Registration Center Setup}
    \label{fig:setup}
\end{figure*}

\begin{figure*}
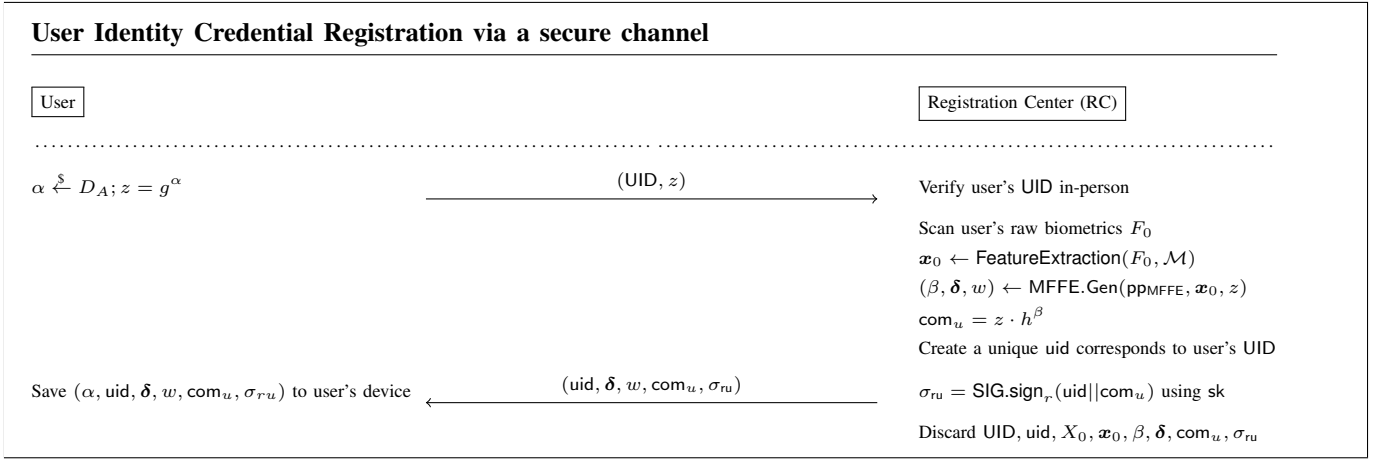

    \centering
    \begin{mdframed}
    	\procedure[codesize=\scriptsize]{\textbf{User Identity Credential Registration via a secure channel}}{
            \\
			\framebox{User} \< \< \< \< \framebox{Registration Center (RC)}   
            \pclb
		    \pcintertext [dotted] {}
            \alpha \sample D_A; z = \ga
             \< \< \sendmessageright*[6cm]{(\UID, z)} \< \< \text{Verify user's } \textsf{UID} \text{ in-person}\\
             \< \< \< \< \text{Scan user's raw biometrics } F_0 \\
            \< \< \< \< \bm{x}_0 \leftarrow \textsf{FeatureExtraction}(F_0, \M) \\
			\< \< \< \< (\beta,\bde, w)\gets \mffe.\Gen(\pp,\bx_0,z) \\
			 \< \<  \< \<  \comu = z \cdot h^\beta \\
             \< \< \< \<  \text{Create a unique } \uid \text{ corresponds to user's } \UID  \\
             \text{Save } (\alpha, \uid, \bde, w, \comu, \sigma_{ru}) \text{ to user's device}  \< \< \sendmessageleft*[6cm]{(\uid, \bde, w, \comu, \sigma_{\ru})} \< \< \sigma_{\ru} = \textsf{SIG.sign}_{r}(\uid||\comu) \text{ using } \sk\\
             \< \< \< \< \text{Discard } \UID, \uid, X_0, \bx_0, \beta, \bde, \comu, \sigma_{\ru}
			}
	\end{mdframed}
    \caption{User Identity Credential Registration. Note that RC does not know $\alpha$ and only needs $z=\ga$ to run $\Gen$.}
    \label{fig:registration_user}
\end{figure*}
\begin{figure*}
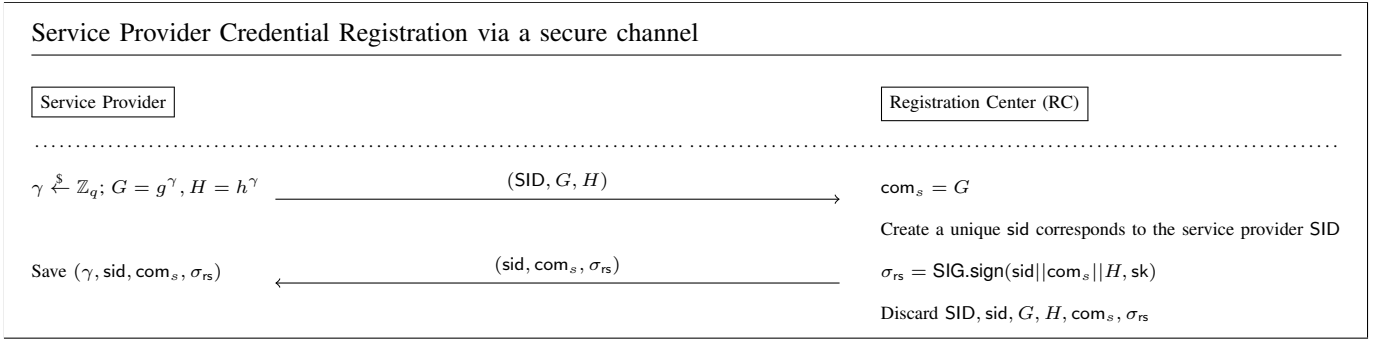

    \centering
    \begin{mdframed}
    	\procedure[codesize=\scriptsize]{Service Provider Credential Registration via a secure channel}{
            \\
			\framebox{Service Provider} \< \< \< \< \framebox{Registration Center (RC)}   
            \pclb
		    \pcintertext [dotted] {}
            \gamma \sample \Zq
            \text{; } G = g^\gamma, H = h^\gamma
             \< \< \sendmessageright*[7.5cm]{(\SID, G, H)} \< \< 
			 \coms = G\\
             \< \< \< \<  \text{Create a unique } \sid \text{ corresponds to the service provider } \SID  \\
             \text{Save } (\gamma, \sid, \coms, \sigma_{\rs}) \< \< \sendmessageleft*[7.5cm]{(\sid, \coms, \sigma_{\rs})} \< \< \sigma_{\rs} = \textsf{SIG.sign}(\sid||\coms||H, \sk)\\
             \< \< \< \< \text{Discard } \SID, \sid, G, H, \coms, \sigma_{\rs} 
			}
	\end{mdframed}
    \caption{Service Provider Credential Registration}
    \label{fig:registration_server}
\end{figure*}

\begin{figure*}
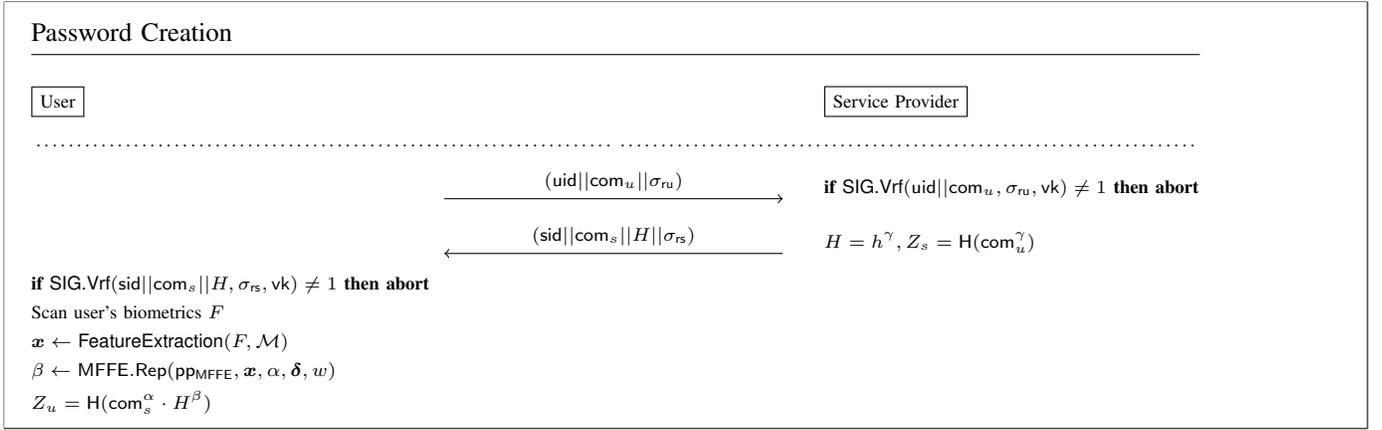

    \centering
    \begin{mdframed}
    	\procedure[codesize=\scriptsize]{Password Creation}{
            \\
			\framebox{User} \< \< \< \< \framebox{Service Provider} \pclb
		    \pcintertext [dotted] {}
            \< \< \sendmessageright*[4.5cm]{(\uid || \comu || \sigma_{\ru})} \< \< \textbf{if } \textsf{SIG.Vrf}(\uid||\comu, \sigma_{\ru}, \vk) \neq 1 \textbf{ then abort} \\ 
            \< \< \sendmessageleft*[4.5cm]{(\sid||\coms||H||\sigma_{\rs})}\< \< H = h^\gamma, Z_s = \hash(\comu^\gamma)  \\
            \textbf{if } \textsf{SIG.Vrf}(\sid||\coms||H, \sigma_{\rs}, \vk) \neq 1 \textbf{ then abort} \\
            \text{Scan user's biometrics } F \\
            \bx \leftarrow \textsf{FeatureExtraction}(F, \mathcal{M}) \\
			\beta \gets \mffe.\Rep(\pp, \bx, \alpha, \bde, w) \\
            Z_u = \hash(\coms^{\alpha} \cdot H^{\beta}) 
            }
	\end{mdframed}
    \caption{Password Creation. Note that the user must know the secret $\alpha$ and have $\bx \in \ARL(\bx_0)$ to recover the key $\beta$.}
    \label{fig:password}
 \end{figure*}

 \begin{figure*}
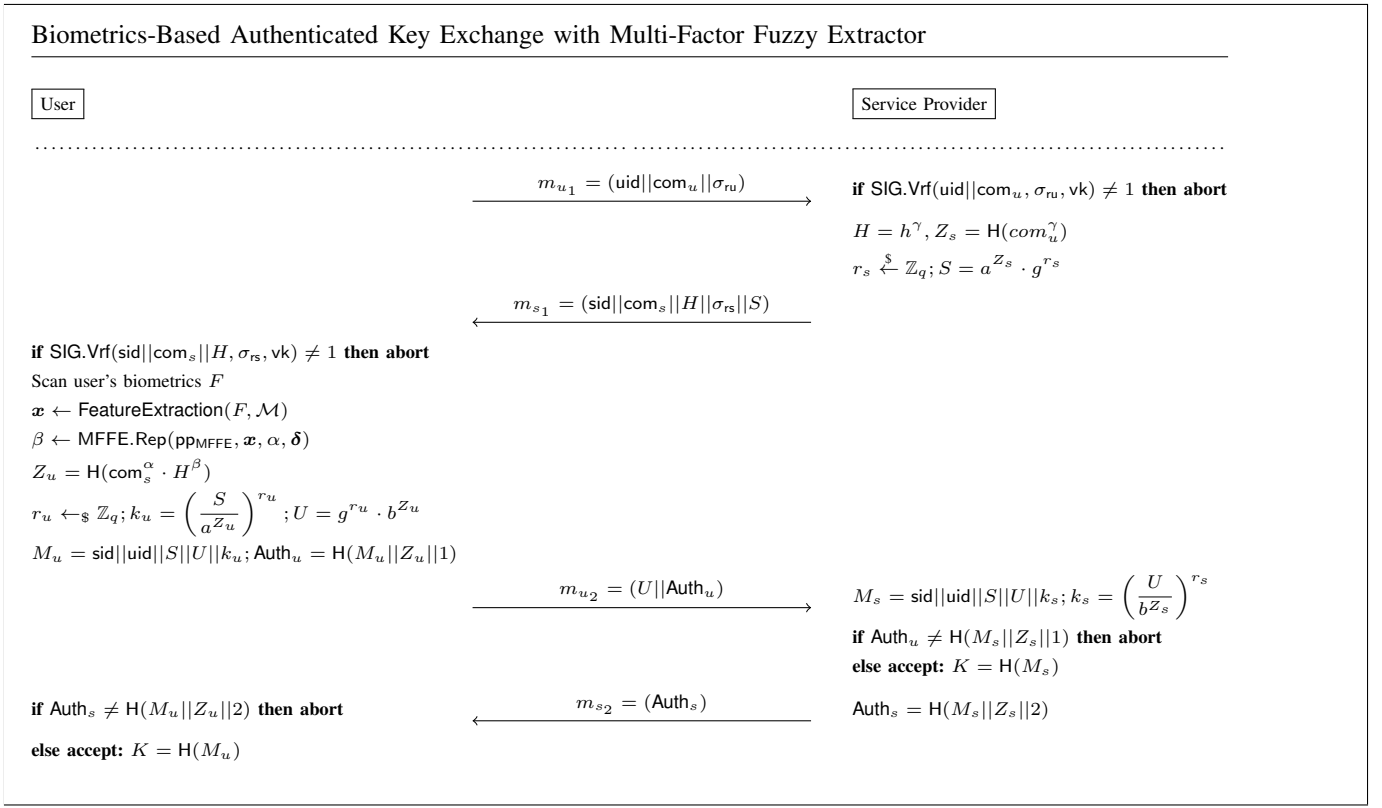

    \centering
    \begin{mdframed}
    	\procedure[codesize=\scriptsize]{Biometrics-Based Authenticated Key Exchange with Multi-Factor Fuzzy Extractor}{
            \\
			\framebox{User} \< \< \< \< \framebox{Service Provider} \pclb
		    \pcintertext [dotted] {}
            \< \< \sendmessageright*[4.5cm]{m_{u_1}=(\uid || \comu || \sigma_{\ru})} \< \< \textbf{if } \textsf{SIG.Vrf}(\uid||\comu, \sigma_{\ru}, \vk) \neq 1 \textbf{ then abort} \\ 
            \< \< \< \< H = h^\gamma, Z_s = \hash(com_u^\gamma)  \\
            \< \< \< \< r_s \sample \mathbb{Z}_q; S = a^{Z_s} \cdot g^{r_s} \\
            \< \< \sendmessageleft*[4.5cm]{m_{s_1} = (\sid||\coms||H||\sigma_{\rs}||S)}\< \< \\
            \textbf{if } \textsf{SIG.Vrf}(\sid||\coms||H, \sigma_{\rs}, \vk) \neq 1 \textbf{ then abort} \\
            \text{Scan user's biometrics } F \\
            \bx \leftarrow \textsf{FeatureExtraction}(F, \mathcal{M}) \\
			\beta \gets \mffe.\Rep(\pp, \bx, \alpha, \bde) \\
             Z_u = \hash(\coms^{\alpha} \cdot H^{\beta})   \\
             r_u \leftarrow_\$ \mathbb{Z}_q; 
            k_u = \left ( \dfrac{S}{a^{Z_u}}\right )^{r_u}; U = g^{r_u} \cdot b^{Z_u} \\
            M_u = \sid||\uid||S||U||k_u;
            \textsf{Auth}_u = \hash(M_u||Z_u||1)\\
             \< \< \sendmessageright*[4.5cm]{m_{u_2} = (U||\textsf{Auth}_u)} \< \< M_s = \sid||\uid||S||U||k_s; k_s = \left ( \dfrac{U}{b^{Z_s}}\right )^{r_s}\\
            \< \< \< \<  \textbf{if }  \textsf{Auth}_u \neq \hash(M_s||Z_s||1) \textbf{ then abort}\\
            \< \< \< \< \textbf{else accept: } K = \hash(M_s) \\
            \textbf{if }  \textsf{Auth}_s \neq \hash(M_u||Z_u||2) \textbf{ then abort}  \< \< \sendmessageleft*[4.5cm]{m_{s_2} = (\textsf{Auth}_s)} \< \< \textsf{Auth}_s = \hash(M_s||Z_s||2) \\
            \textbf{else accept: } K = \hash(M_u) \\
            }
	\end{mdframed}
    \caption{Biometrics-Based Authenticated Key Exchange with Multi-Factor Fuzzy Extractor. The user must have both a valid biometric sample and the correct secret/password in order to successfully authenticate with a service provider.} 
    \label{fig:mfake}
 \end{figure*}
 

\subsection{Description of the proposed scheme} \label{sec:mfake-description}

The proposed scheme consists of three main phases: 1)~\textit{Setup} - by RC to setup system parameters, 2)~ \textit{Registration} - users and service providers register to RC, and 3)~\textit{Authentication and Key Agreement} - by which
the user proves her identity and establishes a random session key with an SP.

At the \textit{Setup} phase (Fig.~\ref{fig:setup}), RC executes $\RC.\Setup$ to generate its secret key $\sk$ and public parameters $\public$, which includes a fuzzy data setting $\F$ as in Section \ref{FDS-with-Lattice}, a group $\mathbb{G}$ of prime order $q$ in which the Discrete Logarithm and the Computational Diffie-Hellman (CDH) assumption holds, two random group elements $g$ and $h$, biometrics feature extraction model $\M$, and the verification key $vk$. RC also runs $\mffe.\Setup$ to generate public parameters $\pp$ of a multi-factor fuzzy extractor as in Section \ref{subsec:mffe-lattice}.

\textit{Registration} phase includes user identity credential registration (Fig.~\ref{fig:registration_user}) and service provider's identity credential registration (Fig. \ref{fig:registration_server}). Users have to be present in person at RC’s secure location with their physical identity documents (e.g., passport). RC verifies users in person and scans user's biometrics $\bx_0$. Then, RC generates a key $\beta$ and a sketch $\bde$ by running $\mffe.\Gen$ with $\bx_0$ and the binding of user's secret/password $g^\alpha$ as input. Next, RC computes and signs the user's identity credential $\comu$ using its secret key $\sk$. Finally, the sketch and the user's signed identifier commitment were written on the user's device. All the user's retrieved data (including the user's biometrics) and the generated data (including $\beta$) are discarded. Fig.~\ref{fig:registration_server} illustrates the process of registration of an SP to RC. At the end of the process, the SP has its credential $\coms$ signed by RC.

In the \textit{Authentication and Key Agreement} phase (Fig.~\ref{fig:mfake}), a user and a server are authenticated to each other. First, they need to create a shared secret $Z_u=Z_s$, which then can be used in a conventional password-based authenticated key exchange (PAKE) (e.g., \cite{abdalla2005simple}). Fig. \ref{fig:password} illustrates the steps needed from both sides to create this shared secret. The user sends her identity credential with RC's signature to the server. After checking the validity of RC's signature for the integrity of the user's identity credential, the server computes the shared secret as $Z_s = \hash(\comu^\gamma)$, then sends to the user its identity credential with RC's signature. If the signature passes the user's checking then the user proceeds with the computation of the shared secret from her side. She needs to have both a biometric sample $\bx$ \textit{close to} $\bx_0$, which was the sample taken in the Registration phase, and the \textit{exact} secret/password $\alpha$ to reproduce the key $\beta$ from the multi-factor feature extractor, and then to compute $Z_u = \hash(\coms^\alpha \cdot H^\beta)$.  
At this step, both the user and the server are successful in computing their shared \textit{long-term secret}, which is $Z_s = Z_u = \textsf{CDH}(\comu, \coms)$. With this shared secret, any off-the-shelf PAKE protocol, e.g. SPAKE~ \cite{abdalla2005simple}, can then be employed to generate a session key for a secure channel between the user and the SP. 

When both biometrics and secret/password are lost, the user can repeat the registration process with the RC using a fresh secret/password $\alpha$ to obtain a new key and credential $\comu$.  
In this case, her corrupted credential $\comu$ must be added to a public revocation list by the RC. Such a list can be regularly downloaded and used by service providers to avoid authenticating corrupted user credentials.

\subsection{Analysis of the proposed scheme} \label{sec:mfake-analysis}

In this section we formally prove the correctness and the semantic security of the proposed multi-factor AKE protocol.

\subsubsection{Correctness}

Theorem~\ref{theorem-correctness} establishes the correctness of our AKE protocol.

\begin{theorem}[Correctness] \label{theorem-correctness}
    In an execution of the protocol $\mathcal{P}$ described in Section \ref{sec:mfake-description}, every valid user, who provides a correct secret $\alpha$ and a biometrics $\bx \hspace{-2pt}\in\hspace{-2pt} \AR(\bx_0)$, is authenticated and successfully establishes a random session key with a service~provider.
\end{theorem}
\begin{proof}
According to the correctness of MFFE (Proposition \ref{pro:mffe_correctness}), the user recovers $\beta$ successfully given the correct $\alpha$ and the biometric $\bx \in \ARL(\bx_0)$. Then she can compute the shared long-term key with the SP, which is $Z_u = \hash(\coms^\alpha H^\beta) = \hash((g^\alpha h^\beta)^\gamma )= \hash(\comu^\gamma) = Z_s$, and the Diffie-Hellman key~$k_u$:
\begin{align*}
    k_u &= \left(\frac{S}{a^{Z_u}}\right)^{r_u} = \left(\frac{g^{r_s} \cdot a^{Z_s}} {a^{Z_u}}\right)^{r_u} = g^{r_sr_u}.
\end{align*}
From the server side, the SP computes the Diffie-Hellman key $k_s$, which is exactly equal to $k_u$
\begin{align*}
    k_s &= \left( \dfrac{U}{b^{Z_s}} \right )^{r_s}  = \left( \dfrac{g^{r_u} \cdot b^{Z_u}}{b^{Z_s} }\right )^{r_s} = g^{r_ur_s}.
\end{align*}
Due to $k_u = k_s = g^{r_ur_s}$ and $Z_u = Z_s$, we have
\begin{align*}
    \textsf{Auth}_u &= \hash(\sid||\uid||S||U||k_u||Z_u||1),\\
    \textsf{Auth}_s &= \hash(\sid||\uid||S||U||k_s||Z_s||2).
\end{align*}
Thus, both the user and the server are authenticated and successfully establish the random session key $K = \hash(\sid||\uid||S||U||
k_u) = \hash(\sid||\uid||S||U||
k_s)$.
\end{proof}

\subsubsection{Security} 

We follow the security model for authenticated key exchange of Bellare \et~\cite{bellare2000authenticated} and Pointcheval and Zimmer~\cite{pointcheval2008multi}. In our model, the adversary can corrupt either the biometric or the secret. 
The goal is to show that in either case, the generated session key should still remain semantically secure.

\textbf{Semantic security}. 
The semantic security of the session key is modeled using the Real-or-Random indistinguishability framework, following Pointcheval and Zimmer~\cite[Sec. 2.2]{pointcheval2008multi}). 
Essentially, this security notion
ensures that an adversary who is capable of eavesdropping, intercepting, forwarding, modifying the messages between the parties, and stealing one authentication factor of an user, can only gain a \textit{negligible} advantage in guessing whether the provided value is a real key or a random value compared to another adversary who performs a random guess.

In Theorem~\ref{theorem-security}, we provide an upper bound on the adversary's advantage $\AdvMfake$.
We adopt a standard game-based approach in which a series of games are constructed and analyzed, where $\textsf{Game}_0$ is the real attack game against the protocol. We keep track of the changes in the adversary's success probability from each game to the next, which eventually allow us to establish an upper bound on the adversary's advantage in winning the security game.
Note that the theorem and its proof cover both the cases when the adversary impersonates the client and the adversary impersonates the server, hence establishing the semantic security for mutual authentication. 

Let $\SuccSig$ denote the probability of the adversary's success in breaking the security of the signature scheme within a time bound $T$, $\SuccCDH$ the probability of the adversary's success in solving the computational Diffie-Hellman problem within a time bound $T$, $\SuccDLU/\SuccDL$ the probability of the adversary's success in solving the discrete logarithm problem for a secret sampled from a distribution $U/D_A$, and $\SuccDLsketch$ the probability of the adversary's success in solving the discrete logarithm problem with sketch for a secret sampled from a distribution $D_B$.
Here, $\GG$ denotes the underlying cyclic group.  

\begin{theorem}
\label{theorem-security}
    Consider the protocol $\mathcal{P}$ described in Section~\ref{sec:mfake-description}
    over a group $\mathbb{G}$ of prime order $q$, where the user secret is drawn from the distribution $D_A$ and the user biometric data is drawn from the distribution $D_X$. 
    Let $\mathcal{A}$ be an adversary against the security of the protocol within a time bound $T$, 
    executing at most $q_s$ sessions of $\mathcal{P}$
    and at most $q_h$ random oracle queries. Then we have
\begin{align*}
\AdvMfake &\leq \frac{q_s^2}{q^2} + \dfrac{q_h^2}{q} + 2 \SuccSig \\
&+ 2q_h^2 \SuccCDH + 2q_h\SuccCDH \\
&+ 2\SuccDLU +2\SuccDL \\
&+ 2\SuccDLsketch + 2\FMR,
\end{align*}
\end{theorem}
\begin{proof}
    Due to the page limit, the proof of this theorem is provided in the supplemental material.
\end{proof}

\subsubsection{Performance} 

In this section, we analyze the performance of our proposed AKE protocol in terms of the computation and communication overhead.

\textbf{Computation}. The computational performance of our protocol is demonstrated in Table~\ref{tab:comp}. 
The group $\GG$ is an elliptic curve group of order $q$.
The computation costs at the user and the server are calculated based on the computation time of elliptic curve addition $T_\boxplus$
and point multiplication $T_\boxtimes$, biometric feature extraction $T_\textsf{FE}$, signature creation/verification $T_\textsf{Sgn}/T_\textsf{Vrf}$, and the running time $T_\textsf{Gen}$ of $\mffe.\Gen$, $T_\textsf{Rep}$ of $\mffe.\Rep$, and the computation time  $T_{\hash}$ of $\hash$. 

As can be seen from Table~\ref{tab:comp-comparison}, the protocol in~\cite{pointcheval2008multi} requires a linear number of elliptic curve point multiplications on both user and server sides, which is very expensive.
The computation overhead of our scheme is on par with that of \cite{gunasinghe2017privbiomtauth} and \cite{Zhang2019}. Note that with additional terms $T_\textsf{Vrf}$ and $T_{\hash}$ (which are of negligible cost), our scheme guarantees mutual authentication, which was not provided in the protocols of~\cite{gunasinghe2017privbiomtauth, Zhang2019}. Also, $T_\textsf{FuEx}$ for the fuzzy extractor in \cite{Zhang2019} can be considered equivalent to $T_\textsf{Rep}$ in our scheme.
As demonstrated by the empirical evaluation in Section~\ref{sec:experiments}, our protocol took only 0.93 seconds on averaged computational time for authenticated key exchange.

\textbf{Communication}. 
Let $\ell_i, \ell_q, \ell_e, \ell_s$ denote the sizes in bits of an identity ($\textsf{uid, sid}$), a $\mathbb{Z}_q$ element, an elliptic curve point, and a signature, respectively. Then the overall communication overhead of our protocol in an authenticated key exchange session is $5\ell_e + 2\ell_q + 2\ell_s + 2\ell_i$. As demonstrated in Table~\ref{tab:comm-comparison}, the communication overhead of the protocol in \cite{pointcheval2008multi} is very expensive, which is $(3n + 4)\ell_e$, much more greater than ours. The communication overhead of the protocol in \cite{gunasinghe2017privbiomtauth} is $5\ell_e + \ell_q + \ell_{\mathsf{IDT}} \geq 5\ell_e + \ell_q + (\ell_q + \ell_s + \ell_i) = 5\ell_e + 2\ell_q + \ell_s + \ell_i$. Compared to \cite{gunasinghe2017privbiomtauth}, the communication overhead of our AKE just adds an insignificant size of one signature and one identity string for server authentication which was not provided in \cite{gunasinghe2017privbiomtauth}. If only client authentication is provided in our protocol as in \cite{gunasinghe2017privbiomtauth}, then our communication time is less than that of \cite{gunasinghe2017privbiomtauth}, just $5\ell_e + \ell_q + \ell_s + \ell_i$. Similarly, the communication overhead of our protocol is also on par with that of \cite{Zhang2019}. With the experimental parameter settings in Section~\ref{sec:experiments}, the communication overhead of our AKE is  $(5\ell_e + 2\ell_q + 2\ell_s + 2\ell_i)/8 = 448$ bytes for mutual authentication and $(5\ell_e + \ell_q + \ell_s + \ell_i)/8 = 344$ bytes for unilateral client authentication, while the communication overheads of the protocols developed in~\cite{pointcheval2008multi,gunasinghe2017privbiomtauth,Zhang2019} are 147648 bytes, 392 bytes, 432 bytes, respectively. 

\begin{table*}
\centering
  \caption{The theoretical analysis of computation costs at the user side and the server side based on the computation time of elliptic curve addition `$T_\boxplus$', elliptic curve point multiplication `$T_\boxtimes$', biometric feature vector extraction `$T_\mathsf{FE}$', signature creation/verification `$T_\textsf{Sgn}, T_\textsf{Vrf}$', and the running time of $\mffe.\Gen$ `$T_\mathsf{Gen}$', $\mffe.\Rep$ `$T_\mathsf{Rep}$', and the time `$T_\hash'$ of $\hash$.}
  \label{tab:comp}
  \begin{tabular}{c|c|c}
    \toprule
    \textbf{Phases} & \textbf{User} & \textbf{RC/SP}\\
    \midrule
    Registration & $T_\boxtimes$ & $T_{\textsf{FE}} + T_{\textsf{Gen}}  + T_\textsf{Sgn} + T_\boxtimes + T_\boxplus$\\
    Authenticated Key Exchange & $T_{\textsf{Vrf}} + T_{\textsf{FE}} + T_{\textsf{Rep}} + 6T_\boxtimes + 3T_\boxplus +  4T_{\hash}$ & $T_{\textsf{Vrf}} + 5T_\boxtimes + 2T_\boxplus +  4T_\hash$\\
  \bottomrule
\end{tabular}
\end{table*}

\begin{table*}[]
\centering
\caption{Comparison of computational time of our authenticated key exchange protocol with the related works. We denote $T_\textsf{Cls}, T_\mathsf{PBKDF}, T_\mathsf{SDec}$ the computation time of classification, password based key derivation function, and symmetric decryption of the trained model in \cite{gunasinghe2017privbiomtauth}, $T_\textsf{ADec}$ and $T_\textsf{FuEx}$ the computation time of asymmetric decryption of user data and fuzzy extractor in \cite{Zhang2019}, respectively. The estimated computation costs exclude $T_\mathsf{FE}, T_\mathsf{Cls}, T_\mathsf{PBKDF}, T_\mathsf{SDec}, T_\mathsf{FuFx}, T_\mathsf{ADec}, T_\mathsf{Rep}$.}
\label{tab:comp-comparison}
\begin{tabular}{@{}c|c|c|c@{}}
\toprule
\textbf{Works}                                               & \textbf{User}                                                                                                                    & \textbf{Service provider}  & \textbf{Estimated comp. cost in ms}                                                      \\ \midrule
Pointcheval and Zimmer \cite{pointcheval2008multi}                                       & $T_{\textsf{FE}} + (2n+4)T_\boxtimes + nT_{\hash}$                                                             & $(2n+4)T_{\boxtimes}$      &       $151,037.44$                                                    \\ \midrule
Gunasinghe and Bertino\cite{gunasinghe2017privbiomtauth}  & $T_{\textsf{FE}} + T_{\textsf{Cls}} + T_{\mathsf{PBKDF}} + T_{\mathsf{SDec}} + 4T_{\boxtimes} + 2T_{\boxplus} + 2T_{\hash}$                                                       & $T_{\textsf{Vrf}} + 7T_{\boxtimes} + 3T_{\boxplus}$  &   $407.92$                                                 \\ \midrule
Zhang~\et\cite{Zhang2019}                   & $T_{\textsf{FE}} + T_{\textsf{FuEx}} + 6T_{\boxtimes} + 2T_\boxplus + 2T_{\textsf{H}}$                                                                                               & $T_{\textsf{ADec}} + 8T_{\boxtimes} + 2T_\boxplus + 2T_{\hash}$ &                       $515.66$       \\ \midrule
\textbf{Our protocol}                                         & $T_{\textsf{Vrf}} + T_{\textsf{FE}} + T_{\textsf{Rep}} + 6T_\boxtimes + 3T_\boxplus +  4T_{\hash}$ & $T_{\textsf{Vrf}} + 5T_\boxtimes + 2T_\boxplus +  4T_\hash$ & $410.58/407.96$ \\ \bottomrule
\end{tabular}
\end{table*}

\begin{table*}
\centering
  \caption{Comparison of communication cost of our authenticated key exchange with the related works. Let $\ell_i, \ell_q, \ell_e, \ell_s, \ell_m, \ell_\mathsf{IDT}$,  denote the sizes in bits of an identity ($\mathsf{uid, sid}$), a $\mathbb{Z}_q$ element, an elliptic curve point, a signature, a MAC message in \cite{Zhang2019}, and an IDT message in \cite{gunasinghe2017privbiomtauth}, respectively, where $\ell_\mathsf{IDT} \geq \ell_i + \ell_e + \ell_s$.}
  \label{tab:comm-comparison}
  \begin{tabular}{c|c|c}
    \toprule
        \textbf{Works} & \textbf{Comm. complexity} & \textbf{Estimated comm. cost in bytes}\\ \midrule
        Pointcheval and Zimmer \cite{pointcheval2008multi} & $(3n+4)\ell_e$ & $147,648$\\ \midrule
        Gunasinghe and Bertino \cite{gunasinghe2017privbiomtauth} & $5\ell_e + \ell_q + \ell_\mathsf{IDT}$ & $392$\\ \midrule
        Zhang~\et \cite{Zhang2019} & $6\ell_e + 2\ell_q + 2\ell_m + 2\ell_i$ & $432$ \\ \midrule
        \textbf{Our protocol} & $5\ell_e + 2\ell_q + 2\ell_s + 2\ell_i$ & $488/344$ \\ 
    \bottomrule
    \end{tabular}
\end{table*}

\begin{table*}
  \begin{minipage}{.5\textwidth}
  \centering
  \caption{Averaged computation time measurements in milliseconds of the basic operations in our implementation.}
  \label{tab:time-operations}
  \begin{tabular}{c|c|c|c|c|c|c|c}
    \toprule
    $T_\hash$ & $T_\mathsf{Sgn}$ & $T_\mathsf{Vrf}$ & $T_\boxtimes$ & $T_\boxplus$ & $T_\mathsf{FE}$ &$T_\mathsf{Gen}$ & $T_\mathsf{Rep}$\\
    \midrule
    $0.01$ & $0.8$ & $2.6$ & $36.8$ & $0.1$ & $190$ & $150$ & $280$\\
    \bottomrule
    \end{tabular}
  \end{minipage}%
  \hspace{0.5cm}
  \begin{minipage}{.5\textwidth}
  \centering
  \caption{Averaged computation time of MFAKE in seconds.}
    \label{tab:time-mfake}
  \begin{tabular}{c|c|c}
    \toprule
    \textbf{Phases} & \textbf{User} & \textbf{RC/SP}\\
    \midrule
    Registration & $0.03$ & $0.4$\\
    Authenticated Key Exchange & $0.74$ & $0.18$\\
    \bottomrule
    \end{tabular}
  \end{minipage}
\end{table*}

\subsection{Related Literature on Multi-Factor Authenticated Key Exchange} \label{sec:mfake-relatedworks}

Multi-factor authenticated key exchange schemes have been designed in a number of works, where multiple factors such as biometric data, secrets, and tokens are combined to provide stronger authentication \cite{Ometov2018}. 
Pointcheval and Zimmer \cite{pointcheval2008multi} proposed a multi-factor authenticated key exchange protocol that uses three factors - a password, a secret, and a biometric sample - to generate a key and establish a secure channel. The matching was performed based on a threshold of the distance between the candidate template and a reference template, both in bit strings. However, there exist some limitations in this protocol. 
The reference biometric template is assumed to be an $n$-bit string and is encrypted bit-wise with ElGamal encryption, which incurs a significant computation and communication overhead. Besides, no experimental evaluation was provided.

The authenticated key exchange protocols developed in  \cite{Yoon2013, He2015, Odelu2015, wazid2017secure, Zhang2019, ma2022outsider, kumar2023robust} requires a third-party (e.g., a registration center/identity authority, a gateway, a database server) to authenticate users-service providers and establish a session key between them. 
This would generate a heavy load on the registration centre, making it expensive to build a scalable and reliable system. Moreover, a compromised registration centre will render all keys generated insecure.

Gunasinghe and Bertino \cite{gunasinghe2017privbiomtauth} proposed a remote authentication scheme in  which users can authenticate to remote services using their biometrics. Their approach is based on a cryptographic
identity token that encodes a user's biometric identifier and a user's password which then can be used in an identification protocol and key agreement. In each user's device, there is a machine learning model, which was built for her biometrics by an identity authority in the enrollment process. In the authentication phase, this machine learning model predicts the user's biometric class label which then is used as a component of the user's identifier. Although their protocol can prevent user identity misbinding, it cannot resist user impersonation from a privileged insider as the identity authority has both the user's biometrics and the user's password. Besides, their protocol does not provide mutual authentication, making it susceptible to the outsider key compromise impersonation attack. An outsider attacker can force a user to generate any session key. 

Zhang~\et~\cite{Zhang2019} proposed a multi-factor authenticated key exchange scheme utilising a fuzzy extractor. Their scheme is not mutually authenticated and susceptible to an outsider key compromise impersonation attack as pointed out by Ma and He in \cite{ma2022outsider}. However, the protocol that Ma and He introduced requires a database server as a third-party and it was not experimentally validated.


Multi-factor authenticated key exchange was introduced to solve the significant drawback of biometrics spoofing attacks and permanent identity loss in single-biometrics factor schemes. Although many schemes were proposed, none of them has addressed all these issues simultaneously as pointed out in Table \ref{tab:literature-comparison}. Most of the proposed multi-factor authentication schemes in the literature are not tolerate attacks arising from malicious insiders and simultaneously prevent user identity misbinding, for example, 
\cite{barman2018provably, kumari2018provably, yang2018cryptanalysis, roy2018provably, zhou2020authentication, chuang2021cake}.
Many remote authentication protocols with biometrics have been theoretically analysed but often lack experimental evaluation of authentication accuracy on a real biometric dataset.

\section{Empirical evaluation} \label{sec:experiments}
We used the SDUMLA finger vein dataset~\cite{yin2011sdumla} (see Fig.~\ref{fig:fingervein} for samples)
to evaluate the empirical performance of our AKE protocol.
The dataset contains 3,816 finger veins from 106 individuals, each of whom provided 36 finger vein images of their 6 fingers (left index, left middle, left ring, right index, right middle, right ring). 

\begin{figure}
    \centering
    \includegraphics[scale=0.5]{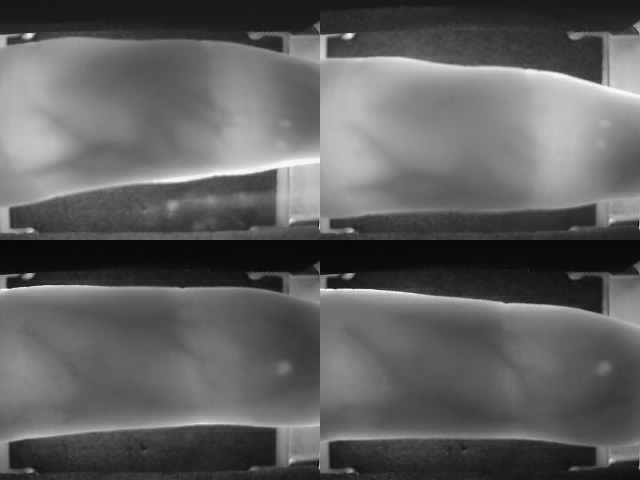}
    \caption{Two samples of two different (left) index finger veins from the SDUMLA dataset~\cite{yin2011sdumla}. The left and the right samples are from the same finger of the same individual. The top and the bottom ones are from two different people.}
    \label{fig:fingervein}
    \vspace{-10pt}
\end{figure}

We implemented 
our protocol in Python
using the libraries \texttt{pytorch}, \texttt{py\_ecc}, \texttt{hashlib}, \texttt{cryptography}, and \texttt{ecdsa}. 
The experiments were run on a 
Macbook laptop with 2 GHz Quad-Core Intel Core i5 and 8GB system memory. 

We used the state-of-the-art Elliptic curve BLS12-381, which is widely used in blockchain applications like Zcash and Ethereum 2.0. We implemented the function $\hash$ as hashing a bit string using \textsf{SHA3-256} to generate a 256-bit digest, which in turn is taken modulo a 255-bit number $q$ (the subgroup order of BLS12-381). 
The encryption and signature schemes employed are \textsf{AES} and \textsf{ECDSA}, respectively. Table~\ref{tab:time-operations} lists the averaged computation time measurements of the basic operations used in our protocol.

The raw finger vein images were fed into a feature extractor model to produce feature vectors in $\mathbb{R}^n$, which are compatible with our proposed scheme. In our experiments, the CNN-based deep learning model in \cite{kuzu2021loss} was employed, resulting in feature vectors in $\mathbb{R}^n$ with $n = 1024$. This learning model effectively increases the Euclidean distance of impostor pairs and decreases that of genuine pairs. The average time of the biometric feature extractor is  $T_\mathsf{FE} = 0.19\text{s}$.

We used the triangular lattice $\TriLatticeB$ spanned by a basis $\bm{B} \in \mathbb{R}^{n \times n}$ with the basis length $d$ (see Section \ref{FDS-with-Lattice}). $\TriLatticeB$ uniquely defines the acceptance region $\textsf{AR}$. A larger basis length $d$ results in a larger acceptance region $\textsf{AR}$. Thus, increasing the basis length $d$ has the effect of increasing $\mathsf{FMR}$ and decreasing $\mathsf{FNMR}$, and vice versa (see Fig.~\ref{fig:d_far_frr.png}).
In our experiment, the basis length $d = 0.254$ provides us with the optimal acceptance region $\AR$ for achieving the equal error rate EER = 0.04\%. 
 
\begin{figure}
    \centering
    \includegraphics[scale=0.6]{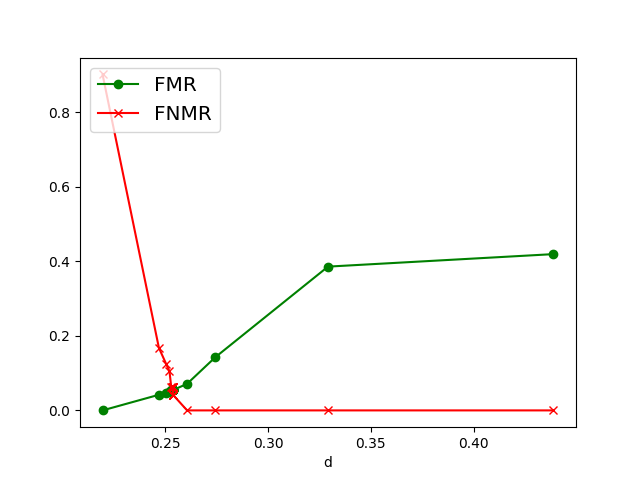}
    \caption{The 
    False Matching Rate (FMR) and False None Matching Rate (FNMR) (in \%) w.r.t. the basis length $d$.}
    \label{fig:d_far_frr.png}
    \vspace{-10pt}
\end{figure}

The average time of $\mffe.\Setup, \mffe.\Gen, \mffe.\Rep$ with different values of $n$ are given in Fig. \ref{fig:time-mffe}. In our experiment with $n=1024$, $\mffe.\Setup = 0.13\text{s}, \mffe.\Gen \hspace{-2pt}=\hspace{-2pt} 0.15\text{s}, \mffe.\Rep \hspace{-2pt} =\hspace{-2pt}0.28\text{s}$. 

\begin{figure}
    \centering
    \includegraphics[scale=0.48]{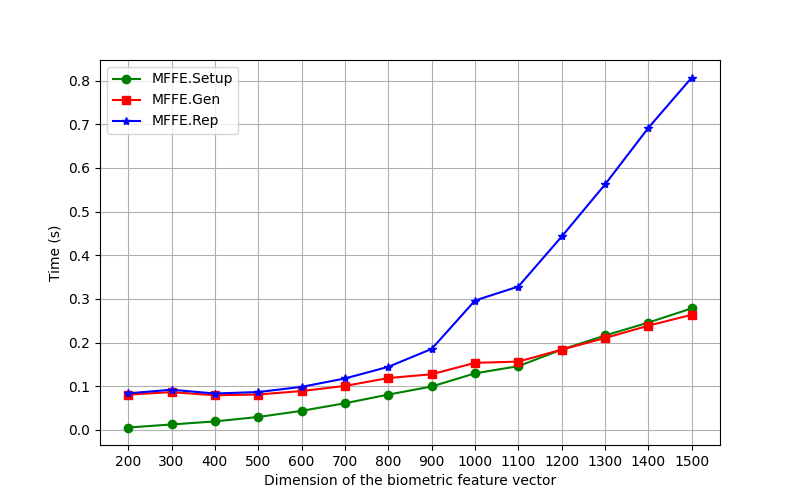}
    \caption{The average running times over 100 executions of $\mffe.\Setup, \mffe.\Gen, \mffe.\Rep$ with different feature vector lengths $n$.}
    \label{fig:time-mffe}
    \vspace{-10pt}
\end{figure}

We conducted experiments to measure the average time of our authenticated key exchange protocol at both the user and the service provider sides over $28620$ tests ($4770$ genuine pairs and $23850$ impostor pairs).
In the registration phase, the average time at the user side was just $0.03$ seconds. On the service provider side, the average time was $0.4$ seconds, mainly due to feature extraction and $\mffe.\Gen$. During the authenticated key exchange phase, the computation time on the user side was $0.74$ seconds, primarily due to feature extraction and $\mffe.\Rep$. On the other hand, the server time in the authenticated key exchange phase was $0.18$ seconds, mostly from the elliptic curve point multiplications. 

Table~\ref{tab:comp-comparison}, Table \ref{tab:comm-comparison} and Fig.~\ref{fig:comm-comp-comparison} compare our protocol with the most related works \cite{pointcheval2008multi,gunasinghe2017privbiomtauth,Zhang2019} in terms of communication and computation overheads. We estimated the computation cost of other protocols based on the averaged computation time measurements of the basic cryptographic operations in Table~\ref{tab:time-operations}. Note that, we do \textit{not} include the feature extractor operation time $T_\mathsf{FE}$, the classification time $T_\mathsf{Cls}$, the fuzzy extractor time $T_\mathsf{FuEx}$ in Zhang~\et~and Gunasinghe and Bertino protocols as they used specific algorithms which cannot be reproduced due to the lack of information in their papers. We also do not count the password based key derivation function time $T_\mathsf{PKCS}$, the decryption time of the trained model $T_\mathsf{SDec}$ in the computation cost of Gunasinghe and Bertino's protocol and the decryption time of user data $T_\mathsf{ADec}$ in the Zhang~\et's protocol. Even if so, our computation cost is comparable to theirs and much less than Pointcheval and Zimmer's protocol.

The communication overhead of our protocol (Table~\ref{tab:comm-comparison}), with the parameters  
$\ell_i = 64, \ell_e = 384, \ell_q = 256, \ell_m = 256, \ell_s = 512$, is only $(5\ell_e + 2\ell_q + 2\ell_s + 2\ell_i)/8 = 448$ bytes for mutual authentication and $(5\ell_e + \ell_q + \ell_s + \ell_i)/8 = 344$ bytes for unilateral client authentication. The communication overhead in the protocols of Poitcheval and Zimmer \cite{pointcheval2008multi}, Gunasinghe and Bertino  \cite{gunasinghe2017privbiomtauth}, and Zhang~\et \cite{Zhang2019}, which only provide unilateral client authentication, are $147648$ bytes, $392$ bytes, and $432$ bytes, respectively. 

These experimental results demonstrate that our proposed biometrics-based AKE is efficient, with communication and computation costs comparable to the state-of-the-art protocols.

\begin{figure}
    \centering
    \includegraphics[scale=0.45]{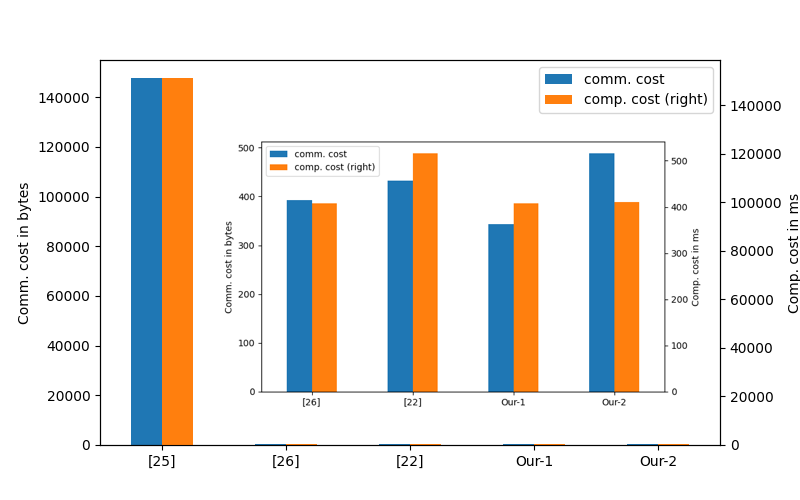}
    \caption{Performance comparison of our protocol with the most related ones. Our-1 and Our-2 correspond to our protocols with mutual authentication and unilateral client authentication, respectively.}
    \label{fig:comm-comp-comparison}
\end{figure}

\section{Conclusion}
\label{sec:conclusion}
To enhance security in biometrics-based authentication systems, we propose a novel multi-factor fuzzy extractor that integrates both a user's secret (e.g. a password) and a user's biometrics in the generation and reconstruction process of a cryptographic key. 
We employ this multi-factor fuzzy extractor to construct a new multi-factor authenticated key exchange protocol that possesses desirable features: first, the protocol provides mutual authentication; second, the user and service provide can authenticate each other without the involvement of the identity authority; third, the protocol can prevent user impersonation from a compromised authority; and finally, the protocol is reusable even when both biometric sample and secret are leaked. Most existing works on multi-factor AKE only have a subset of these features. We also formally prove that the proposed AKE protocol is semantically secure. 

We evaluated our AKE protocol on a real finger vein dataset and achieved a low equal error rate of 0.04\%. Our protocol was also reasonably fast, taking an average of only 0.93 seconds of computation time for the user and service provider to authenticate and establish a shared session key. Its communication overhead is only 448 bytes, which is negligible.

\section*{Acknowledgments}
The authors would like to thank Wataru Nakamura and other co-authors of~\cite{Katsumata2021} for very helpful discussions. 



\bibliographystyle{IEEEtran}
\bibliography{IEEEabrv,bioauth.bib}

\newpage

 

\begin{IEEEbiography}[{\includegraphics[width=1in,height=1.25in,clip,keepaspectratio]{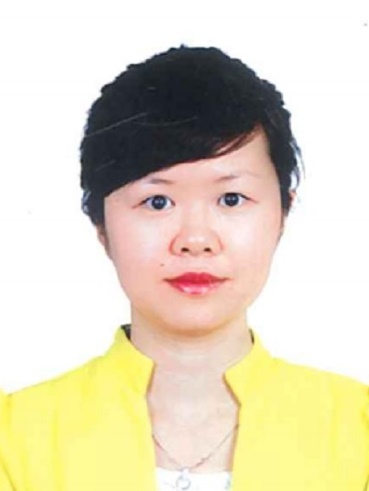}}]{Hong-Yen Tran}
	is currently a Research Associate with the School of Systems and Computing, University of New South Wales, Canberra, Australia. Her research interests are in the field of applied cryptography, data security and privacy, and bio-cryptography.
\end{IEEEbiography}

\vspace{11pt}

\begin{IEEEbiography}[{\includegraphics[width=1in,height=1.25in,clip,keepaspectratio]{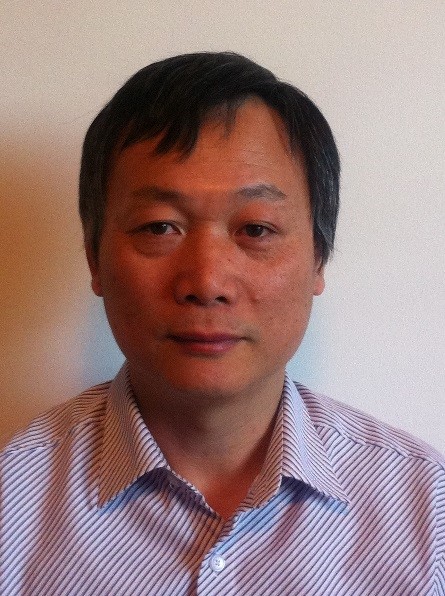}}]{Jiankun Hu}
	is currently a Professor with the School of Systems and Computing, University of New South Wales, Canberra, Australia. He is also an invited expert of Australia Attorney-General’s Office, assisting the draft of Australia National Identity Management Policy. He has received nine Australian Research Council (ARC) Grants and has served at the Panel on Mathematics, Information, and Computing Sciences, Australian Research Council ERA (The Excellence in Research for Australia) Evaluation Committee 2012. His research interests are in the field of cyber security covering intrusion detection, sensor key management, and biometrics authentication. He has many publications in top venues, including the IEEE TRANSACTIONS ON PATTERN ANALYSIS AND MACHINE INTELLIGENCE, the IEEE TRANSACTION COMPUTERS, the IEEE TRANSACTIONS ON PARALLEL AND DISTRIBUTED SYSTEMS, the IEEE TRANSACTIONS ON INFORMATION FORENSICS AND SECURITY, Pattern Recognition, and the IEEE TRANSACTIONS ON INDUSTRIAL INFORMATICS. He is a senior area editor of the IEEE TRANSACTIONS ON INFORMATION FORENSICS AND SECURITY.
\end{IEEEbiography}

\vspace{11pt}

\begin{IEEEbiography}[{\includegraphics[width=1in,height=1.25in,clip,keepaspectratio]{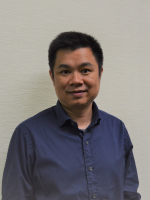}}]{Wen Hu}
Wen Hu is a professor at School of Computer Science and Engineering, the University of New South Wales (UNSW). Much of his research career has focused on the novel applications, low-power communications, security, signal processing and machine learning Cyber Physical Systems (CPS) and Internet of Things (IoT). Hu published regularly in the top rated sensor network and mobile computing venues such as ACM/IEEE IPSN, ACM SenSys, ACM MobiCOM, ACM UbiCOMP, IEEE PerCOM, ACM TOSN, IEEE TMC, IEEE TIFS and IEEE TDSC.

Hu was a principal research scientist and research project leader at CSIRO Digital Productivity Flagship, and received his Ph.D from the UNSW. He is a recipient of prestigious CSIRO Office of Chief Executive (OCE) Julius Career Award (2012 - 2015) and multiple research grants from Australian Research Council, CSIRO and industries.

Hu is a senior member of ACM and IEEE, and is an associate editor of ACM TOSN, the general chair of CPS-IoT Week 2020, co-chair the program committee of ACM/IEEE IPSN 2023 and ACM Web Conference (WWW 2023, Systems and Infrastructure for Web, Mobile Web, and Web of Things track), as well as serves on the organising and program committees of networking conferences including ACM/IEEE IPSN, ACM SenSys, ACM SigCOMM, ACM MobiCOM and ACM MobiSys.

Hu actively commercialises his research results in smart buildings and IoT, and his endeavours include working as the Chief Technical Officer (part time) and the Chief Scientist (part time) in Parking Spotz (2021 - 2022) and WBS Tech (2016-2020) respectively.
\end{IEEEbiography}

\clearpage

{\large\textbf{Supplemental Material: Proof of Theorem~\ref{theorem-security}}}.

We follow the security model for authenticated key exchange of Bellare \et~\cite{bellare2000authenticated} and Pointcheval and Zimmer~\cite{pointcheval2008multi}. In our model, the adversary can corrupt either the biometric or the secret. 
The goal is to show that in either case, the generated session key should still remain semantically secure.

We first define the necessary notation in the model. We assume that servers and clients (users) can activate several instances at a time to run several sessions simultaneously. In this model, the instance $i$ of the party $P$, denoted $\Pi^i_P$, where $P$ can be a user or a server, includes:
\begin{itemize}
    \item $\textsf{pid}^i_P$: the instance partner that $\Pi^i_P$ believes it is interacting with,
    \item $\sid^i_P$: the session identifier, which is the session transcript seen by $\Pi^i_P$, 
    \item $\textsf{acc}^i_P$: a boolean variable denotes if the instance $\Pi^i_P$ goes in an accepted state at the end of the session.
\end{itemize}
The two instances $\Pi^i_P$ and $\Pi^j_{P'}$
are said to be partners if the following conditions are fulfilled:
\begin{itemize}
    \item $\textsf{pid}^i_P = \Pi^j_{P'}$ and $\textsf{pid}^j_{P'} = \Pi^i_{P}$;
    \item $\sid^i_P = \sid^j_{P'} \neq \textsf{NULL}$;
    \item $\textsf{acc}^i_P = \textsf{acc}^j_{P'} = \texttt{True}$.
\end{itemize}

\textbf{Adversaries' capabilities}. 
The interaction between an adversary $\mathcal{A}$ and the protocol participants is simulated via oracle queries, which also models the adversary's capabilities 
in a real attack. 
The oracle query types available to the adversary are as follows:
\begin{itemize}
    \item $\textsf{Execute}(\Pi^i_U, \Pi^j_S)$: This oracle models the capability of the adversary in eavesdropping the messages over the transmission network. This oracle outputs the transcript of an honest execution of the protocol between $\Pi^i_U$ and $\Pi^j_S$.
    \item $\textsf{Send}(m, \Pi^i_P)$: This oracle models the capability of the adversary in intercepting, forwarding, modifying or creating the message $m$ sent to the instance $\Pi^i_P$. This oracle outputs the answer of $\Pi^i_P$ to the message $m$.
    \item $\textsf{Corrupt}(U, a)$: This query models the capability of the adversary in breaking authentication factors of a user~$U$. This oracle outputs $\alpha$ if $a = 1$, and $\bx \in \ARL(\bx_0)$ if $a = 2$. 
\end{itemize}

A user $U$ is \textit{fully corrupted} by an adversary if all $\textsf{Corrupt}$ queries targeting $U$, namely $\textsf{Corrupt}(U, 1)$ and $\textsf{Corrupt}(U, 2)$, have been called. 
All $\textsf{Corrupt}$ queries must be executed before a session starts. 
An instance $\Pi^i_P$ is \textit{fresh} if upon accepted $\Pi^i_P$ is not fully corrupted.
A session key is \textit{fresh} if both participants are fresh instances.

\textbf{Semantic security}. 
The semantic security of the session key is modeled using the Real-or-Random indistinguishability framework, following Pointcheval and Zimmer~\cite[Sec. 2.2]{pointcheval2008multi}). 
Essentially, this security notion
ensures that an adversary who is capable of eavesdropping, intercepting, forwarding, modifying the messages between the parties, and stealing one authentication factor of an user, can only gain a \textit{negligible} advantage in guessing whether the provided value is a real key or a random value compared to another adversary who performs a random guess.

\textbf{Security Game for the Protocol $\mathcal{P}$}
\label{def:attackgame}
At the beginning of the game, the challenger chooses a bit $b$. The adversary can ask the challenger $\mathsf{Test}(\Pi^i_P)$ queries. If $\Pi^i_P$ is not fresh, the challenger outputs $\perp$, otherwise, the challenger answers the adversary with either real keys generated by $\Pi^i_P$ honestly following protocol $\mathcal{P}$ if $b = 1$ or random keys if $b = 0$. At the end of the game, the adversary outputs a bit $b'$. If $b' = b$ then the adversary wins the games, otherwise, she loses.

Let $\bm{S}$ be the event that the adversary $\mathcal{A}$ correctly guesses the bit $b$ used by the challenger during the above attack game. The \textit{advantage} of $\mathcal{A}$ in winning the game is defined as
\begin{align}
    \label{eq:advmfake}
    \AdvMfake &\triangleq \Pr[b' = 1|b = 1] - \Pr[b'=1|b=0]   \nonumber \\ 
    &= \Pr[b' = 1|b = 1] - (1 - \Pr[b'=0|b=0])  \nonumber \\ 
    &= 2 \Pr[\bm{S}] - 1.
\end{align}

In Theorem~\ref{theorem-security}, we provide an upper bound on the adversary's advantage $\AdvMfake$.
We adopt a standard game-based approach in which a series of games are constructed and analyzed, where $\textsf{Game}_0$ is the real attack game against the protocol. We keep track of the changes in the adversary's success probability from each game to the next, which eventually allow us to establish an upper bound on the adversary's advantage in winning the security game.
Note that the theorem and its proof cover both the cases when the adversary impersonates the client and the adversary impersonates the server, hence establishing the semantic security for mutual authentication. 

Let $\SuccSig$ denote the probability of the adversary's success in breaking the security of the signature scheme within a time bound $T$, $\SuccCDH$ the probability of the adversary's success in solving the computational Diffie-Hellman problem within a time bound $T$, $\SuccDLU/\SuccDL$ the probability of the adversary's success in solving the discrete logarithm problem for a secret sampled from a distribution $U/D_A$, and $\SuccDLsketch$ the probability of the adversary's success in solving the discrete logarithm problem with sketch for a secret sampled from a distribution $D_B$.
Here, $\GG$ denotes the underlying cyclic group.  


\begin{proof}[Proof of Theorem~\ref{theorem-security}]

The proof consists of a sequence of games, starting with the real game $\game_0$ and ending with $\game_6$. In each $\game_i$, we consider the events $\bm{S}_i$ of $\mathcal{A}$ successes that occurs if the adversary $\mathcal{A}$ correctly guesses the bit $b$ chosen by the challenger at the beginning of $\game_i$.
Let $\Delta_i \triangleq \Pr[\bm{S}_i]-\Pr[\bm{S}_{i+1}]$ denote the distance between the probabilities of the adversary's successes in $\game_i$ and $\game_{i+1}$.

$\textbf{\game}_0$: The real attack game of $\mathcal{A}$ against protocol $\mathcal{P}$. 

$\textbf{\game}_1$: We simulate the hash oracle $\hash$ as 
a random oracle with a list $\Lambda_\hash$ according to the following rule. 
For a hash query $\hash(\q)$, if $(\q, \r)$ appears in $\Lambda_\hash$ then the answer is $\r$. Otherwise, a new random $\r$ is set to be the answer and $(\q, \r)$ is added to $\Lambda_\hash$. 

All the instances for $\mathsf{Execute, Send, Reveal, Corrupt}$ and $\mathsf{Test}$ queries are simulated as the real players would do in the protocol. $(m_{u_1}||m_{s_1}||m_{u_2}||m_{s_2})$ is added to the map $\Lambda_\mathsf{T}$ to keep track of the session transcripts. As $\game_1$ and $\game_0$ are indistinguishable in the random oracle model, we have $\Delta_0 = 0$.

$\textbf{\game}_2$: In this game, we cancel the games in which the adversary succeeds in forging the RC's signature. As $\game_2$ and $\game_1$ only differ if the RC's signature is forged, the distance between the probabilities of the adversary's success in these two games is bounded by the probability of the adversary's successes in breaking the security of the signature scheme within the time bound $T$, which is $\SuccSig$. Hence,
\begin{align} 
    \Delta_1 &\leq \SuccSig. \label{eq:deltaS1} 
\end{align}

$\textbf{\game}_3$: To guarantee the independence of the sessions, from $\game_2$ we cancel the games in which collisions on the session transcripts
occur. Since there is at least one honest party, $S||\mathsf{Auth}_s$ or $U||\mathsf{Auth}_u$ is uniformly distributed. As we cancel the games in which the adversary succeeds in forging the RC’s signature from $\game_2$, the probability of the collision on the session transcript is bounded from above by the probability of choosing the same $S||\mathsf{Auth}_s \in \GG \times \Zq$ or the same $U||\mathsf{Auth}_u \in \GG \times \Zq$ in some two sessions among $q_s$ sessions. As $|\GG|=|\Zq|=q$, they are bounded from above by $\frac{q_s^2}{2q^2}$. This can be achieved by either using the birthday paradox or by using a simple union bound\footnote{Consider $q_s$ i.i.d. uniformly distributed pairs $(a_i,b_i)\in \Zq^2$. Let $E$ denote the event that some two among $q_s$ pairs are the same, and $E_{i,j}$ the event that $(a_i,b_i)\equiv (a_j,b_j)$. Then $E = \cup_{1\leq i\neq j \leq q_s}E_{i,j}$, and hence, by the union bound, $\Pr[E]\leq \sum_{1\leq i\neq j \leq q_s}\Pr[E_{i,j}] = \binom{q_s}{2}\frac{1}{q^2}<\frac{q_s^2}{2q^2}$.}. 
As $\game_1$ and $\game_2$ only differ if the collisions of session transcripts happen, the distance between the probabilities of the adversary's successes in these two games is bounded by the probability of such collisions. Therefore, 
\begin{align} 
    \Delta_2 &\leq \frac{q_s^2}{2q^2}. \label{eq:deltaS2} 
\end{align}

$\textbf{\game}_4$: In this game, we avoid the collisions amongst the hash queries asked by $\mathcal{A}$ to $\hash$ by aborting the game if the answer $\r$ of the random oracle $\hash$ to the query $\q$ of the adversary $\mathcal{A}$ is found in $\Lambda_\mathcal{A}$.

As $\game_2$ and $\game_3$ only differ if a hash collision occurs among $q_h$ hash queries, the distance between the probabilities of the adversary's successes in these two games is bounded by the probability of such collision. Therefore, using either the union bound or the birthday paradox as in $\game_2$, we derive that
\begin{align} 
    \Delta_3 &\leq \dfrac{q_h^2}{2q}. \label{eq:deltaS3}
\end{align}

$\textbf{\game}_5$: From $\game_4$, for all fresh sessions, we replace the generations of the authenticators $\mathsf{Auth}_u, \mathsf{Auth}_s$ and the session key $K$ with a private random oracle $\hash'$ so that the values of $\mathsf{Auth}_u, \mathsf{Auth}_s, K$ are independent of $\hash$ and $Z_u, Z_s, k_u, k_s$. Thus we can omit the computation of $Z_u, Z_s, k_u, k_s$
and set $S\hspace{-2pt} =\hspace{-2pt} g^{s^*}, U \hspace{-2pt}=\hspace{-2pt} g^{u^*}$ with $s^*, u^* \hspace{-2pt}\sample\hspace{-2pt} \mathbb{Z}_q$. 

Denote $\mathsf{AskH5}$ the event that the adversary asks the \textit{correct} $(\mathsf{Auth}_u, M_u)$ or $(\mathsf{Auth}_s, M_s)$ to the $\hash$ oracle while $\hash'$ has been used by the challenger. 
Note that $\game_4$ and $\game_5$ return the same outcome when $\mathsf{AskH5}$ does \textit{not} happen. Therefore,
\begin{align} \label{eq:deltaS4}
    \Delta_4 &\leq \Pr[\mathsf{AskH5}] .
\end{align}


Note that, in $\mathsf{Game}_5$, because the session keys $K$ are computed from a private random oracle $\hash'$, they appear to $\mathcal{A}$ as random. 
Therefore, 
\begin{align} 
    \Pr[\bm{S}_5] &= \dfrac{1}{2}. \label{eq:S5} 
\end{align}

$\textbf{\game}_6$: 
From $\game_5$, we cancel the games in which, for some transcript generated before any $\mathsf{Corrupt}$ query is made and coming from an execution involving $\mathcal{A}$ impersonates a participant (the user or the server), there are two tuples $\left(S, U, \mathsf{CDH}\left(\frac{S}{a^{Z_k}}, \frac{U}{b^{Z_k}} \right)\right)$ corresponding to two different non-zero $Z_k$ ($k = 0, 1$), such that
$\left(\sid||\uid||S||U||\mathsf{CDH}\left(\frac{S}{a^{Z_k}}, \frac{U}{b^{Z_k}} \right)||Z_k||1\right)$ 
(if $\mathcal{A}$ impersonates the user) or ($\sid||\uid||S||U||\mathsf{CDH}\left(\frac{S}{a^{Z_k}}, \frac{U}{b^{Z_k}} \right)||Z_k||2$) (if $\mathcal{A}$ impersonates the server)
are in the map $\Lambda_\hash$. \\

If such a pair exists, we have
\begin{align}
    \mathsf{CDH} (a, b) &= \dfrac{ \mathsf{CDH} \left(\dfrac{S}{a^{Z_0}}, \dfrac{U}{b^{Z_0}} \right) ^ \frac{1}{Z_0(Z_1-Z_0)}}{\mathsf{CDH} \left(\dfrac{S}{a^{Z_1}}, \dfrac{U}{b^{Z_1}} \right) ^ \frac{1}{Z_1(Z_1-Z_0)}} \times U^\frac{-\mathsf{DL}(S)} {Z_0 Z_1} \label{eq:cdh1} \\
    &= \dfrac{ \mathsf{CDH} \left(\dfrac{S}{a^{Z_0}}, \dfrac{U}{b^{Z_0}} \right) ^ \frac{1}{Z_0(Z_1-Z_0)}}{\mathsf{CDH} \left(\dfrac{S}{a^{Z_1}}, \dfrac{U}{b^{Z_1}} \right) ^ \frac{1}{Z_1(Z_1-Z_0)}} \times S^\frac{-\mathsf{DL}(U)} {Z_0 Z_1}. \label{eq:cdh2}
\end{align}

As we consider the executions in which $\mathcal{A}$ impersonates the user or $\mathcal{A}$ impersonates the server (i.e., $S$ or $U$ is simulated, respectively), the challenger knows the discrete log of $S$ ($\mathsf{DL}(S)$) or the discrete log of $U$ ($\mathsf{DL}(U)$). The challenger can pick randomly from $\Lambda_\hash$ to get the two tuples $\left(\mathsf{CDH}\left(\frac{S}{a^{Z_0}}, \frac{U}{b^{Z_0}})\right), Z_0\right), \left(\mathsf{CDH}\left(\frac{S}{a^{Z_1}}, \frac{U}{b^{Z_1}})\right), Z_1 \right)$ with probability $\frac{1}{q_h^2}$ to compute $\mathsf{CDH}(a, b)$ according to $(\ref{eq:cdh1}, \ref{eq:cdh2})$. 
\begin{align} 
    \Delta_5 &\leq q_h^2 \SuccCDH. \label{eq:deltaS5}
\end{align}


We denote by $\mathsf{AskH6}$ the same event as $\mathsf{AskH5}$ but in $\game_6$. We divide $\mathsf{AskH6}$ into two types:

    \textit{Type-1}: The event $\mathsf{AskH6}$ with the executions between a user instance and a server instance, denoted by $\mathsf{AskH6passive}$;
    
    \textit{Type-2}: The event $\mathsf{AskH6}$ with the executions between the adversary and a user instance or a server instance with at least one message flow is not oracle-generated (i.e., the adversary impersonates the server), denoted by $\mathsf{AskH6active}$;

    
We have
\begin{align}
    \Pr[\mathsf{AskH5}] &\leq \Pr[\mathsf{AskH6}] + \Delta_5 
\end{align}

Assume there is a tuple $\left(S, U, \mathsf{CDH}\left(\frac{S}{a^Z}, \frac{U}{b^Z} \right)\right)$ such that 
$(\sid||\uid||S||U||\mathsf{CDH}\left(\frac{S}{a^Z}, \frac{U}{b^Z}||Z||1 \right))$ or $(\sid||\uid||S||U||\mathsf{CDH}\left(\frac{S}{a^Z}, \frac{U}{b^Z}||Z||2 \right))$ is in the map $\Lambda_\hash$. We consider the cases as follows:

\begin{itemize}
    \item The corresponding transcript of the tuple $\left(S, U, \mathsf{CDH}\left(\frac{S}{a^Z}, \frac{U}{b^Z} \right)\right)$ is generated from \textit{Type-1} event. 
    
    \hspace{12pt}Because both $S$ and $U$ have been simulated and the discrete log of $S$ ($s^*$) and the discrete log of $U$ ($u^*$) are known, we have
    \begin{align}
        \mathsf{CDH}(a, b) = \left( \dfrac{g^{u^*s^*} \cdot (a^{u^*}b^{s^*})^Z}{\mathsf{CDH}\left(\frac{S}{a^Z}, \frac{U}{b^Z})\right)} \right)^{\frac{1}{Z^2}}.
    \end{align}
    
    \hspace{12pt}The challenger can pick randomly from $\Lambda_\hash$ to get the value of $\mathsf{CDH}\left(\frac{S}{a^Z}, \frac{U}{b^Z})\right)$ with probability $\frac{1}{q_h}$ to compute $\mathsf{CDH}(a, b)$. Therefore, the probability that $\mathcal{A}$ queries $\hash$ at $(\sid||\uid||S||U||\mathsf{CDH}\left(\frac{S}{a^Z}, \frac{U}{b^Z} \right))$ where $Z = \mathsf{CDH}(\comu, \coms)$ in Type-1 event is upper bounded by
        \begin{equation} \label{eq:ask6-passive}
            \Pr[\mathsf{AskH6passive}]  \leq q_h \SuccCDH.
        \end{equation}
    \item The corresponding transcript of the tuple $\left(S, U, \mathsf{CDH}\left(\frac{S}{a^Z}, \frac{U}{b^Z} \right)\right)$ is generated from~\textit{Type-2} event. 
    
    The probability that $\mathcal{A}$ queries $\hash$ at $\left(\sid||\uid||S||U||\mathsf{CDH}\left(\frac{S}{a^Z}, \frac{U}{b^Z} \right)||Z||1\right)$  or $\left(\sid||\uid||S||U||\mathsf{CDH}\left(\frac{S}{a^Z}, \frac{U}{b^Z} \right)||Z||2\right)$ where $Z = \mathsf{CDH}(\comu, \coms)$ is equal to the probability that $\mathcal{A}$ successfully computes $Z$, which is upper bounded by the probability of the adversary's success in solving the discrete logarithm problem for the server's secret $\gamma$ within the time bound $T$ ($\SuccDLU$) and the probability of the adversary's success in obtaining the user's secrets $(\alpha, \beta)$ given one $\mathsf{Corrrupt}$ query. 

    

    We analyze the probability of $\mathcal{A}$'s success in obtaining ($\alpha, \beta$) of $\comu$ to impersonate the user in the following cases of corruption when one of the authentication factors (but not both) is revealed to $\mathcal{A}$:
    \begin{itemize}
    \item $\mathsf{Corrupt}(U, a=1)$ has been made by $\mathcal{A}$, that is, $\mathcal{A}$ knows $\alpha$.\\
     Knowing $\alpha$, $\mathcal{A}$ computes $h^\beta = \dfrac{\comu}{\ga}$.
    
     The probability of $\mathcal{A}$'s success in obtaining $\beta$ includes the probability of $\mathcal{A}$'s success in solving the discrete logarithm problem with sketch $\DLsketch$ for the key $\beta$ within the time bound $T$ and the false matching rate of the biometric space, hence it is equal to $(\SuccDLsketch + \FMR)$. 
     
     \item $\mathsf{Corrupt}(U, a=2)$ has been made by $\mathcal{A}$ (i.e., $\mathcal{A}$ obtains a close biometric sample $\bm{x} \in \ARL(\bx_0)$). Consider $\mathcal{A}$ as a privileged insider adversary, which means $\mathcal{A}$ knows ($\bx_0, \beta, \ga)$ (e.g. a compromised RC leaked these data in the registration process). Because $\alpha$ and $\bx$ are independent, the probability of $\mathcal{A}$'s success in obtaining $\alpha$ is the probability of $\mathcal{A}$'s success in solving the discrete logarithm problem for the secret $\alpha$ within the time bound $T$, which is $\SuccDL$.
    \end{itemize}
\end{itemize}

Therefore, we have
\begin{align} \label{eq:ask6-active}
    \Pr[\mathsf{AskH6active}] 
    &\leq  \SuccDLU + \SuccDL \nonumber \\
    &+ \SuccDLsketch + \FMR .
\end{align}

From (\ref{eq:ask6-passive}) and (\ref{eq:ask6-active}), we have
\begin{align} \label{eq:ask6}
    \Pr[\mathsf{AskH6}] &\leq q_h \SuccCDH  + \SuccDLU + \SuccDL \nonumber \\
    &+ \SuccDLsketch + \FMR .
\end{align}

Finally, the advantage of $\mathcal{A}$ in breaking the semantic security of session keys in the proposed protocol multi-factor authenticated key exchange is:
\begin{align*}
    \AdvMfake &= 2  \mathsf{Pr}[\textbf{S}_0] - 1  \\
    &= 2  (\mathsf{Pr}[\textbf{S}_0] - \mathsf{Pr}[\textbf{S}_5]) + 2  \mathsf{Pr}[\textbf{S}_5] - 1 \\
    &\leq 2  \sum_{i=0}^4{\Delta_i} + 2  \mathsf{Pr}[\textbf{S}_5] - 1 \\
    &= 2  \sum_{i=0}^4{\Delta_i}  \\
    &\leq 2  \sum_{i=0}^3{\Delta_i} + 2  (\Delta_5 + \mathsf{Pr}[\mathsf{AskH6}])
\end{align*}
From $(\ref{eq:deltaS1}), (\ref{eq:deltaS2}), (\ref{eq:deltaS3}), (\ref{eq:deltaS5}), (\ref{eq:ask6})$, we have
\begin{align*}
\AdvMfake &\leq \frac{q_s^2}{q^2} + \dfrac{q_h^2}{q} + 2 \SuccSig \\
&+ 2q_h^2 \SuccCDH + 2q_h\SuccCDH \\
&+ 2\SuccDLU +2\SuccDL \\
&+ 2\SuccDLsketch + 2\FMR, 
\end{align*}
which completes the proof.
\end{proof}

\vfill

\end{document}